\newif\ifproofs
 \proofstrue 

\documentclass[10pt,journal]{IEEEtran}
 
\IEEEoverridecommandlockouts

\usepackage{pgfplotstable}
\usepackage{xstring}
\usepackage{booktabs}
\usepackage{multirow}
\usepackage{adjustbox}
 
\usepackage{standalone}  
\usepackage{graphicx}
\usepackage{csvsimple}
\usepackage{booktabs}
\usepackage{multirow}
\usepackage{array}
\usepackage{arydshln}
\usepackage{adjustbox}

\makeatletter
\newcommand{\ifstreq}[2]{%
  \begingroup
  \edef\temp@a{#1}%
  \edef\temp@b{#2}%
  \ifx\temp@a\temp@b
    \endgroup\expandafter\@firstoftwo
  \else
    \endgroup\expandafter\@secondoftwo
  \fi
}
\makeatother

\usepackage{KalmanNet}
\usepackage{booktabs}
\usepackage{multirow}
\usepackage{array}
\usepackage{arydshln}
  \usepackage[all=normal,paragraphs=tight,floats=normal,mathspacing=normal,wordspacing=tight,charwidths=tight,mathdisplays=normal,leading=normal]{savetrees}

\acrodef{cp}[CP]{conformal prediction}
\acrodef{mlp}[MLP]{multi-layer perceptron}
\acrodef{mdp}[MDP]{Markov decision process}
\acrodef{cgkf}[CGKF]{Conformalized Gaussian \ac{kf}}
\acrodef{cqkf}[CQKF]{Conformalized Quantile \ac{kf}}
\acrodef{cdkf}[CDKF]{Conformalized Distributional \ac{kf}}

\begin{document}

\title{Conformalized Kalman  Filters for State Estimation with  Trustworthy Confidence  Regions
}

\author{Olga Weisman, Nir Shlezinger, and Bracha Laufer-Goldshtein
\thanks{
Parts of this work were presented at the IEEE International Conference on Acoustics, Speech, and Signal Processing (ICASSP) 2026 as the paper \cite{weisman2026conformal}. 
 O. Weisman and N. Shlezinger are with the ECE School, Ben-Gurion University of the Negev, Israel (e-mail: olya.weisman@gmail.com; nirshl@bgu.ac.il). 
 B. Laufer-Goldshtein is with the Department of EE, Tel-Aviv University, Israel (email: blaufer@tauex.tau.ac.il). 
 This work was supported by the Israel Science Foundation (ISF) under grant no. 3314/25, by the European Research Council (ERC) under the ERC starting grant nr. 101163973 (FLAIR), and by the Israeli Ministry of Science and Technology.}
 }

\maketitle


\begin{abstract}
Kalman-type filters are widely used for tracking dynamic systems, yet the confidence regions commonly derived from their estimated covariances can become unreliable under nonlinearities, non-Gaussian disturbances, and model mismatch. In this work, we develop a \ac{cp} framework for equipping Kalman-type filters with statistically reliable confidence regions. Unlike \ac{cp} applied to black-box estimators, our approach exploits the recursively estimated first- and second-order moments of the state posterior, which capture the time-varying uncertainty induced by the underlying dynamics. Based on these statistical features, we propose three complementary constructions. The first conformally calibrates Gaussian confidence regions induced by the filter moments. The second employs quantile regression to map the estimated moments into the boundaries of adaptive convex regions, which are subsequently calibrated. The third learns a Gaussian-mixture representation of the posterior and conformalizes the resulting density-based regions, enabling the characterization of multimodal and non-convex uncertainty sets. We establish finite-sample coverage guarantees for both sample-wise confidence, which controls miscoverage at individual time instances, and trajectory-wise confidence, which jointly covers the state sequence over a prescribed horizon. Numerical experiments across diverse linear and nonlinear dynamic systems demonstrate that the proposed methods attain the prescribed coverage while producing tight and informative confidence regions, and highlight the relative merits of the three constructions under different posterior characteristics.
\end{abstract}

\acresetall

\section{Introduction}
\label{sec:intro}  
The \ac{kf} and its variants constitute one of the most widely used families of algorithms for tracking dynamic systems~\cite{gannot2008kalman}. These methods are based on a \ac{ss} representation of the underlying dynamics~\cite{durbin2012time}. For linear Gaussian \ac{ss} models, the classical \ac{kf} yields the minimum \ac{mse} estimate~\cite{kalman1960new}, while extensions such as the \ac{ekf}~\cite{schmidt1981kalman} and \ac{ukf}~\cite{julier2004unscented} enable its application to  nonlinear settings. In many applications, however, accurate state tracking alone is insufficient. The estimated state is often used by downstream mechanisms for tasks such as detection, control, planning, and decision-making, whose reliability depends not only on the estimated value but also on the {\em uncertainty} associated with it. It is therefore important to accompany the state estimate with trustworthy \emph{confidence regions}, which characterize the range of plausible state values~\cite{dahan2024uncertainty,rozenfeld2025conformal}.

Kalman-type algorithms naturally provide a measure of uncertainty as they recursively track the first- and second-order moments of the state. Accordingly, confidence regions are commonly constructed from the estimated state covariance under a Gaussian approximation of the posterior distribution~\cite{zhou2023real,eichstadt2016evaluation}. This construction is statistically justified for linear Gaussian \ac{ss} models, which underlie the classical derivation of the \ac{kf}~\cite{kalman1960new}. Outside this ideal setting, the propagated covariance may not faithfully characterize the estimation error distribution. Nonlinearities, non-Gaussian disturbances, and mismatch between the assumed and actual dynamics may all cause the resulting confidence regions to be miscalibrated, even when the corresponding point estimates remain accurate~\cite{uhlmann2024gaussianity}. Consequently, conventional Gaussian regions may be either overconfident, failing to attain their prescribed coverage, or excessively conservative and thus uninformative.

In recent years, \emph{\ac{cp}} has emerged as a  statistical framework for constructing prediction {intervals} in a data-driven manner with rigorous coverage guarantees~\cite{vovk2005algorithmic,angelopoulos2023conformal,angelopoulos2024theoretical}. {The central principle of split \ac{cp}} is to use a separate \emph{calibration set} to empirically characterize the errors of a given predictive model  and onstruct prediction sets around its outputs so that the resulting prediction sets attain a prescribed coverage level. This yields finite-sample, distribution-free guarantees under exchangeability, without requiring an explicit probabilistic model for the prediction error~\cite{romano2019conformalized,lei2018distribution,gupta2022nested}. Such guarantees make \ac{cp} particularly appealing when conventional model-based uncertainty estimates may be misspecified.

Much of the \ac{cp} literature is developed for standard classification and scalar regression \ac{ml} problems, while state estimation generally involves a continuous-valued \emph{vector} state. Extending \ac{cp} to multivariate regression requires characterizing joint prediction regions rather than scalar intervals and is therefore considerably more challenging~\cite{kong2012quantile,feldman2023calibrated, paindaveine2011directional, bovcek2017directional,dheur2025unified,rosenberg2022fast,vedula2023continuous,fang2025contra,sadinle2019least}. Existing approaches construct such regions through, e.g., learned directional quantiles~\cite{feldman2023calibrated} or learned representations of the conditional output distribution~\cite{sadinle2019least}. While these methods enable multi-dimensional confidence-region geometries, they typically treat the underlying predictor as a generic regression model, and do not explicitly exploit the structured uncertainty information available in model-based state estimators.

A second relevant body of work addresses \ac{cp} for sequential and non-stationary data~\cite{gibbs2021adaptive,xu2023conformal,zecchin2024localized,angelopoulos2023conformalPID,wang2026online,xu2024conformal}. Online conformal methods adapt their prediction sets as the data distribution evolves, typically by updating the conformal threshold according to the observed coverage of previous predictions. 
This mechanism is suited to settings in which the target associated with each past prediction is subsequently revealed, but is generally incompatible with state tracking, where the latent state remains unobserved after an estimate is produced. Related conformal methodologies have also been considered for trajectory prediction, safe planning, and sequential decision systems~\cite{su2024collaborative,lindemann2023safe,dietterich2022conformal,lindemann2024formal}. These works demonstrate how temporal errors can be calibrated, including through aggregation over an entire trajectory, thereby motivating both \emph{sample-wise} and \emph{trajectory-wise} notions of coverage. However, they generally conformalize prediction residuals or externally generated trajectory forecasts and do not use the recursively evolving uncertainty to shape the confidence regions.

Conformal calibration has further been combined with model-based probabilistic uncertainty in Bayesian learning settings~\cite{stanton2023bayesian,kim2025robust}. Such approaches demonstrate that an approximate probabilistic posterior can provide valuable input-dependent information for constructing more informative \ac{cp} sets, even when the probabilistic model itself is imperfect. Yet, these studies primarily address settings such as Bayesian optimization, where queried outcomes are subsequently observed and can be used for sequential calibration. They do not consider recursive estimation of an unobserved dynamic state, nor the particular statistical information provided by Kalman-type filtering. Thus, existing \ac{cp} methods either treat state estimation as a generic regression problem, rely on online feedback unavailable in tracking, or exploit model-based uncertainty in substantially different settings.

\subsection*{Contributions}
In this work, we develop a \ac{cp} framework tailored to Kalman-type tracking of dynamic systems. Our key observation is that conformalizing a Kalman-type filter differs fundamentally from applying \ac{cp} on top of a black-box \ac{ml} estimator, as Kalman-type filters recursively provide estimates of the first- and second-order moments of the state posterior. Although these moments do not necessarily define an accurate Gaussian posterior, their recursive computation accounts for much of the temporal variation of the dynamic system. We therefore leverage these time-varying statistical features to construct calibrated confidence regions, combining the adaptivity and structural information of Kalman filtering with the statistical reliability of \ac{cp}. Our framework supports confidence regions of varying flexibility, ranging from Gaussian-shaped regions to learned convex and multimodal regions.

Our main contributions are summarized as follows:
\begin{itemize}
    \item {\bf Framework:} We introduce a  framework for augmenting Kalman-type filters with  calibrated confidence regions. The framework uses the estimated first- and second-order posterior moments as time-dependent statistical features, thus exploiting the uncertainty information  propagated by the filter rather than treating the  estimator as a black box.

    \item {\bf Methodology:} We propose three complementary constructions for deriving confidence regions from these  features:
    \begin{enumerate}
        \item {\em \ac{cgkf}}, which retains the Gaussian posterior approximation implicit in \acp{kf} and uses calibration data to correct the confidence regions induced by its estimated first- and second-order moments. This construction is particularly suitable when the state posterior can be reasonably approximated as Gaussian.

        \item {\em \ac{cqkf}}, which maps the estimated first- and second-order moments into the boundaries of a convex confidence region, which is subsequently calibrated using held-out trajectories. This construction requires additional training data, but accommodates asymmetric and non-Gaussian posterior distributions, while remaining primarily suited to unimodal distributions due to the convex form of the resulting region.

        \item {\em \ac{cdkf}}, in which a dedicated model maps the moments produced by the Kalman-type filter into the parameters of a Gaussian-mixture approximation of the state posterior. The  density-based regions are then conformally calibrated, enabling the representation of multimodal posterior distributions and potentially non-convex confidence regions.
    \end{enumerate}

    \item {\bf Theory:} We establish statistical coverage guarantees for the proposed confidence regions under both relevant notions of reliability in dynamic state estimation: \emph{sample-wise} coverage, which controls the miscoverage probability at each time instance, and \emph{trajectory-wise} coverage, which controls the probability that any state along the considered trajectory falls outside its corresponding region.

    \item {\bf Experiments:} We evaluate the proposed conformalized \ac{kf} methods over diverse dynamic systems and operating conditions. Our  results demonstrate that the proposed constructions provide reliable coverage while producing tight and informative confidence regions, and illustrate the relative merits of the Gaussian, quantile-regression, and distribution-based formulations under different  characteristics.
\end{itemize}

The rest of this paper is organized as follows: Section~\ref{sec:System Model and Preliminaries} presents the system model. The conformalized \ac{kf} methods are presented in Section~\ref{sec: Method} and evaluated in Section~\ref{sec: Numerical Study}, while Section~\ref{sec:conclusion} provides concluding remarks.

Throughout this paper, we use  boldface lowercase for vectors, e.g., $\myVec{x}$, and boldface uppercase letters, e.g., $\myMat{M}$, for matrices.  Calligraphic letters, e.g., $\mySet{X}$, are used for sets. 
We use  $\mathbb{R}$ for the set of real numbers and $\mathbb{N}$ for the natural numbers. The operations $(\cdot)^\top$ and  $\| \cdot \|$  are used for transpose and  $\ell_2$ norm,  respectively, while $\mathcal{N}(\cdot,\cdot)$ is the Gaussian distribution.

\section{System Model and Preliminaries}
\label{sec:System Model and Preliminaries}  

\subsection{State-Space Models}
\label{ssec:SSmodel}
\ac{ss} models are statistical representations of dynamic systems. They describe the evolution of a latent state vector $\gvec{s}_t \in \mathbb{R}^m$ and its associated observation vector $\gvec{z}_t \in \mathbb{R}^n$ at each time step $t \in \mathbb{N}$ via equations of the form~\cite{durbin2012time}
\begin{subequations}
    \label{eq:ssModel}
\begin{align}
    \gvec{s}_t &= f_t(\gvec{s}_{t-1}) + \gvec{w}_t, \label{eq:state_transition} \\
    \gvec{z}_t &= h_t(\gvec{s}_t) + \gvec{v}_t. \label{eq:gen_state}
\end{align} 
\end{subequations}
In \eqref{eq:ssModel},  $f_t: \mathbb{R}^m \mapsto \mathbb{R}^m$ 
is the state transition function, and $h_t: \mathbb{R}^m \mapsto \mathbb{R}^n$ is the observation function. The process noise $\mathbf{w}_t$ and observation noise $\mathbf{v}_t$ are mutually and temporally independent, with covariance matrices $\mathbf{Q}_t$ and $\mathbf{R}_t$. 

\subsection{\ac{kf}-Type Tracking of Dynamic Systems} 
\label{ssec:KF}
Kalman-type algorithms are a family of filtering methods which estimate $\gvec{s}_t$ from {current as well as past observations} $\mathbf{z}_{1:t}$~\cite{durbin2012time}. This is achieved by recursively propagating estimates of the first- and second-order moments of $\gvec{s}_t$ given $\mathbf{z}_{1:t}$, denoted {as} $\hat{\gvec{s}}_{t}$ and $\boldsymbol{\Sigma}_{t}$ respectively, in two stages:

$(i)$ The {\em predict} step uses  $\hat{\gvec{s}}_{t-1}$ and $\boldsymbol{\Sigma}_{t-1}$ to estimate the moments of $\gvec{s}_t$ and $\mathbf{z}_t$ given $\mathbf{z}_{1:t-1}$. 
For instance, in the \ac{ekf}, the predict step estimates these moments as
\begin{subequations}
    \label{eqn:EKFpredict}
\begin{align}
    \hat{\gvec{s}}_{t|t-1} &= f_t(\hat{\gvec{s}}_{t-1}), \,      &\boldsymbol{\Sigma}_{t|t-1} = \hat{\mathbf{F}}_t \boldsymbol{\Sigma}_{t-1} \hat{\gvec{F}}_{t}^\top + \gvec{Q}_t, \\
    \hat{\gvec{z}}_{t|t-1} &= h_t(\hat{\gvec{s}}_{t|t-1}),    \,     &  \gvec{S}_{t|t-1} = \hat{\gvec{H}}_t \boldsymbol{\Sigma}_{t|t-1} \hat{\gvec{H}}_t^\top + \gvec{R}_t,
\end{align}
\end{subequations}
where $\hat{\mathbf{F}}_{t}$ and $ \hat{\mathbf{H}}_t$ are computed by the local linearizations, $   \hat{\mathbf{F}}_t = \nabla_{\mathbf{s}} f_{t}(\hat{\gvec{s}}_{t-1})$ and  $\hat{\mathbf{H}}_t = \nabla_{\mathbf{s}} h_t(\hat{\mathbf{s}}_{t|t-1})$, respectively. 

$(ii)$ The {\em update} step forms a linear filter (the {\em Kalman gain}) $\mathbf{K}_t$, and uses it to obtain  $\hat{\mathbf{s}}_{t}$ and $\boldsymbol{\Sigma}_{t}$ by balancing the predicted moments with the current observation $\gvec{z}_t$. For example, the \ac{ekf} sets 
$\mathbf{K}_t = \boldsymbol{\Sigma}_{t|t-1} \hat{\mathbf{H}}_t^\top \mathbf{S}_{t|t-1}^{-1}$, and updates the posterior moments via 
    \begin{subequations}
    \label{eqn:EKFupdate}
    \begin{align}
        \hat{\gvec{s}}_{t} &= \hat{\gvec{s}}_{t|t-1} + \mathbf{K}_t ( \mathbf{z}_t - \hat{\mathbf{z}}_{t|t-1}), \\
        \boldsymbol{\Sigma}_{t} &= \boldsymbol{\Sigma}_{t|t-1} - \mathbf{K}_t \mathbf{S}_{t|t-1} \mathbf{K}_t^\top.
    \end{align}
    \end{subequations}
{When $f_t(\cdot)$ and $h_t(\cdot)$ are linear, the \ac{ekf} reduces to the standard \ac{kf}, which is \ac{mse}-optimal when the process and observation noises, as well as the initial state $\myVec{s}_0$, are Gaussian.}

Alternative \ac{kf}-type algorithms share the above structure, while employing different methods for propagating the moments in \eqref{eqn:EKFpredict}-\eqref{eqn:EKFupdate}. For instance, the \ac{ukf}~\cite{julier2004unscented} replaces the local linearization of the EKF with the unscented transform, which propagates a set of deterministically chosen sigma points through the nonlinear state-transition and observation functions to approximate the posterior mean and covariance; the cubature Kalman filter \cite{arasaratnam2009cubature} employs the cubature rule to approximate the Gaussian-weighted integrals arising in nonlinear Bayesian filtering.

\subsection{Problem Formulation}
\label{sec:Problem Formulation} 
State estimation  often requires not only accurate point-wise estimates, but also reliable uncertainty quantification. For linear Gaussian \ac{ss} models, the posterior state distribution is Gaussian and is thus fully characterized by its first- and second-order moments, which are optimally estimated by the \ac{kf}. In this setting, confidence regions can be constructed analytically from the posterior covariance matrix.
When these assumptions do not hold, characterizing confidence regions is considerably more challenging. 

We seek to enhance Kalman-type filters by providing statistically valid confidence guarantees, even when the assumptions underlying Gaussian filtering are violated. Specifically, for a user-specified miscoverage level $\alpha\in(0,1)$, we aim to construct a confidence region $\mathcal{C}_t^\alpha\subseteq\mathbb{R}^m$ that contains the  state with high probability while remaining as compact as possible. Depending on the application, the  confidence guarantee may be either $(i)$ {\em sample-wise}, ensuring coverage independently at each time instant,
\begin{subequations}
\label{eqn:Confidence}
\begin{equation}
\Pr({\gvec{s}}_t \notin \mySet{C}_t^{\alpha}) \leq \alpha,
\label{eqn:InstConfidence}
\end{equation}
or $(ii)$ {\em trajectory-wise}, ensuring that the entire trajectory is jointly covered with probability at least $1-\alpha$,
\begin{equation} \Pr\bigl({\gvec{s}}_t \in \mySet{C}_t^{\alpha}, \forall t \in \{1,\ldots,T\}\bigr) \geq  1- \alpha. \label{eqn:TrajConfidence} \end{equation}
\end{subequations}

We are specifically interested in confidence-region construction mechanisms that satisfy the following requirements:
\begin{enumerate}[label={C\arabic*}]
\item \label{itm:nonlinear} The state transition and observation functions in the underlying \ac{ss} model may be {\em nonlinear}, and the distribution of the process and observation noises may be {\em non-Gaussian}.
\item \label{itm:Alg} The proposed mechanism should be compatible with a broad family of {\em Kalman-type estimators}.
\item \label{itm:mismatch} The \ac{ss} model available to the estimator may be a {\em mismatched approximation} of the true system dynamics.
\end{enumerate}

To address~\ref{itm:nonlinear}--\ref{itm:mismatch}, we assume access to a set of multiple trajectories generated by the underlying dynamical system, denoted 
\begin{equation}
 \mySet{D}= \big\{ \big\{ (\gvec{s}_t^{(i)},\myVec{z}_t^{(i)}) \big\}_{t=1}^{T} \big\}_{i=1}^{|\mySet{D}|}.
 \label{eqn:dataset}
\end{equation}
Our objective is to construct confidence regions that satisfy either the sample-wise and trajectory-wise guarantees while remaining as tight as possible and fully exploiting the information already produced by Kalman-type filters.

\subsection{Preliminaries: Conformal Prediction}
\label{ssec:Conformal Prediction} 
\Ac{cp} is a statistical method for quantifying uncertainty in predictive \ac{ml} models{~\cite{vovk2005algorithmic,angelopoulos2023conformal,angelopoulos2024theoretical}}. In the context of multivariate regression, it considers a  generic case where data is obtained from an unknown distribution $\Pr(\myVec{x},\myVec{y})$ over an input-output space $\mySet{X}\times \mySet{Y}$, with the objective of predicting an $m$-dimensional  vector $\myVec{y}\in\mySet{Y}=\mathbb{R}^{m}$ from a feature vector {$\myVec{x}\in\mySet{X}\subseteq\mathbb{R}^q$}. 
Given a trained \ac{ml} model and a new input $\myVec{x}$, \ac{cp} constructs a {prediction region} $\mySet{C}^\alpha(\myVec{x}) \subseteq \mySet{Y}$ such that the true outcome lies within {$\mySet{C}^\alpha(\myVec{x})$ with probability at least $1-\alpha$, i.e., $\Pr\!\left(\myVec{y} \in \mySet{C}^{\alpha}(\myVec{x})\right) \geq   1-\alpha$. To construct such sets,} split \ac{cp} requires (in addition to an \ac{ml} model trained on training data $\mySet{D}_{\rm train}$), a calibration dataset $\mySet{D}_{\rm cal}
=\{(\myVec{x}^{(i)},\myVec{y}^{(i)})\}_{i=1}^{|\mySet{D}_{\rm cal}|}$.
{Two \ac{cp} approaches relevant to our setting, illustrated in Fig.~\ref{fig:prediction_regions}, are} \emph{quantile regression}~\cite{kong2012quantile,feldman2023calibrated,paindaveine2011directional,bovcek2017directional}, and \emph{distribution learning}~\cite{sadinle2019least}, reviewed next.



\subsubsection{Quantile Regression}
\label{sssec:DQR} 
\Ac{cp} methods based on quantile regression aim to directly learn the quantile intervals without imposing a statistical model on the posterior~\cite{romano2019conformalized}, by first training a dedicated quantile regression model using $\mySet{D}_{\rm train}$, and then calibrating it using $\mySet{D}_{\rm cal}$. For multivariate regression, \ac{dqr} \cite{kong2012quantile,feldman2023calibrated, paindaveine2011directional, bovcek2017directional} is based on characterizing one-dimensional  intervals over different projections of $\mathcal{Y}$, using a dedicated \ac{dnn} with parameters $\myVec{\theta}$ that is trained using  $\mySet{D}_{\rm train}$ to map an input $\myVec{x}$ and a direction vector $\myVec{u} \in \mathbb{R}^m$ into the $\alpha$-quantile interval limit of $\myVec{u}^\top \myVec{y}$.  Specifically, given a set of unit-norm directions $\mySet{U} \subset \mathbb{R}^m$, \ac{dqr} trains its \ac{dnn} $\hat{\mu}_{\myVec{\theta}}:\mySet{X}\times \mySet{U} \mapsto \mathbb{R}$ to approach
%
 \begin{equation}
\label{eqn:cr_leaning}    
\arg\min_{\myVec{\theta}} 
\sum_{\myVec{u} \in \mySet{U}}
\sum_{(\myVec{x}, \myVec{y}) \in \mySet{D}_{\rm train}}
\rho_\alpha \left(\myVec{u}^\top \myVec{y}, \hat\mu_{\myVec{\theta}}(\myVec{x},\myVec{u})\right),
\end{equation}
with $\rho_\alpha$ being the elementwise pinball loss ~\cite{koenker1978regression}, given by
\begin{equation}
\rho_\alpha(y,\hat y)
=
\begin{cases}
\alpha(y-\hat y), & y>\hat y,\\
(1-\alpha)(\hat y-y), & \text{otherwise}.
\end{cases}
\label{eqn:Pinball}
\end{equation}


The calibration set $\mySet{D}_{\mathrm{cal}}$ is used to construct the prediction region by evaluating the directional quantile model on the calibration samples and computing the nonconformity scores
\begin{equation}
\label{eqn:residuals_DQR}
R^{(i)}  = \max_{u\in\mySet{U}} \left( \hat{\mu}_{\myVec{\theta}} (\myVec{x}^{(i)},\myVec{u}) - \myVec{u}^\top \myVec{y}^{(i)} \right),
\end{equation}
which measure the largest violation of the predicted directional quantile constraints over all directions. These residuals are then used to obtain the prediction region by setting $\hat Q^{1-\alpha}$ to the empirical $\left\lceil
\left(\left|\mySet{D}_{\mathrm{cal}}\right|+1\right)(1-\alpha)
\right\rceil$  quantile of $\{R^{(i)}\}$, resulting in
\begin{equation}
\label{eqn:dqr_cp_region}
\!\!\mathcal{C}^{\alpha}_{\mathrm{DQR}}(\myVec{x}) \!=\! \bigcap_{\myVec{u}\in\mySet{U}}\!
\left\{\myVec{y}\in\mathbb{R}^{m}: \myVec{u}^\top \myVec{y} \ge
\hat{\mu}_{\myVec{\theta}}(\myVec{x},\myVec{u} ) \!-\!\hat Q^{1-\alpha} \right\},
\end{equation}
which is a convex subset of $\mathbb{R}^m$. 
In the scalar case where $m=1$ and $\mySet{U} = \pm 1$, \ac{dqr} produces a confidence region given by
\begin{equation}
    \label{eqn:pred_int_cqr}
\mathcal{C}^{\alpha}_{\mathrm{DQR}}(\myVec{x}) = \left[ \hat{\mu}_{\myVec{\theta}}(\myVec{x}, -1 )- \hat Q^{1-\alpha},\ \hat{\mu}_{\myVec{\theta}}(\myVec{x},1)+ \hat Q^{1-\alpha} \right],
\end{equation}
specializing scalar \ac{cqr}~\cite{romano2019conformalized}.

\begin{figure}
\centering
\includegraphics[width=\columnwidth]{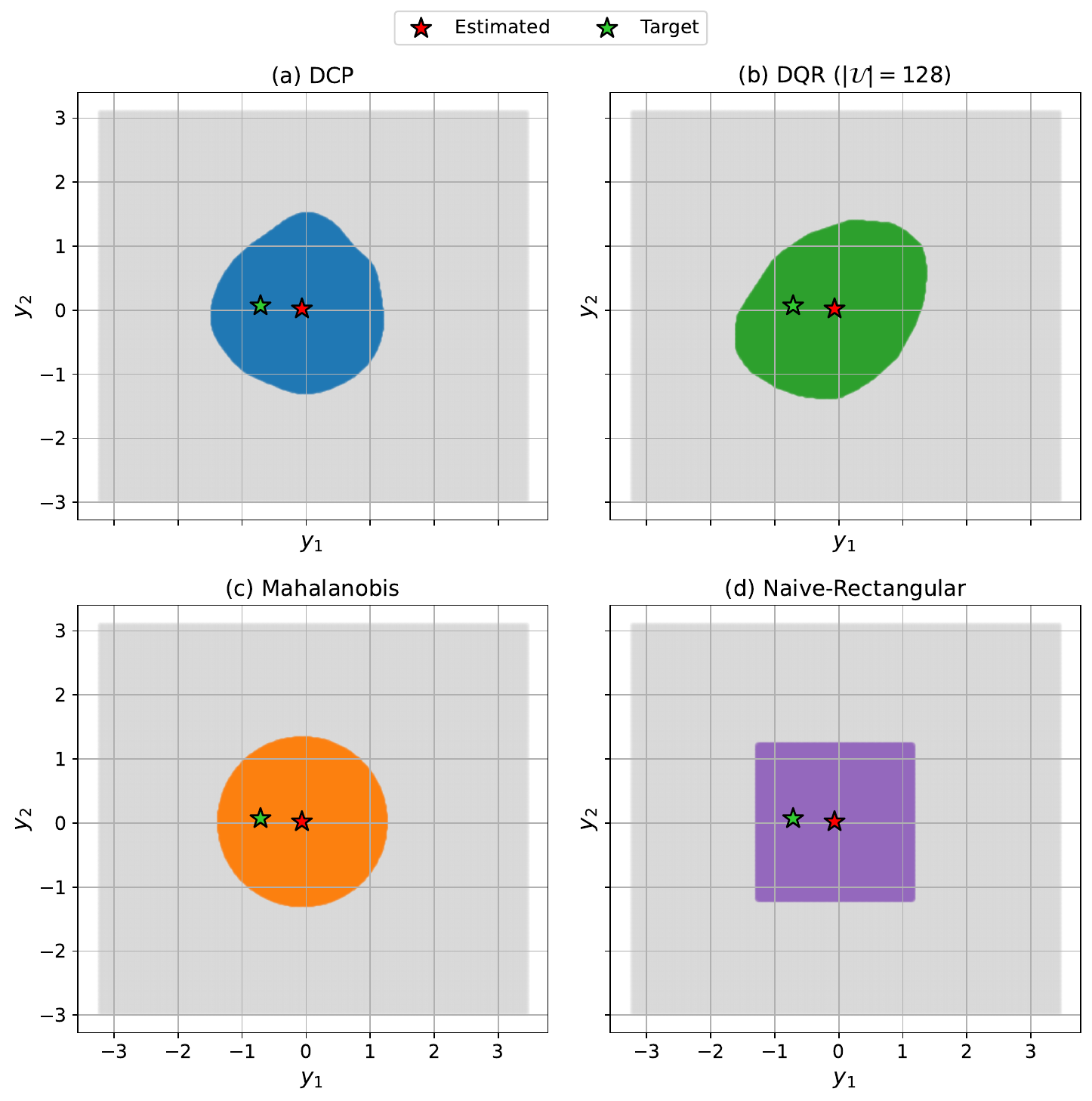}
\caption{Prediction regions illustrations for $m=2$. Upper row: generic \ac{dcp} (left) and \acs{dqr} (right); lower row: special cases of Gaussian (left), which is \ac{dcp} with $K=1$ and rectangular (right), i.e., \acs{dqr} with $\mySet{U} = \{[\pm 1, 0]^\top, [0, \pm 1]^\top\}$. The estimated and target states are shown for reference. }
\label{fig:prediction_regions}
\end{figure}
\subsubsection{Distributional \ac{cp}}
\label{sssec:DR-CP} 
Distributional methods utilize  a conditional density estimator $\hat{f}_{\hat{\myVec{\theta}}}(\myVec{y}|\myVec{x})$. A notable approach is \ac{dcp}~\cite{sadinle2019least}, which constructs prediction regions from a Gaussian mixture model for the posterior, given by
\begin{equation}
\hat{f}_{\hat{\myVec{\theta}}}(\myVec{y}|\myVec{x})=
\sum_{k=1}^{K}
\pi_k(\myVec{x})
\mathcal{N}\!\left(y;\myVec{\mu}_k(\myVec{x}),\myMat{\Sigma}_k(\myVec{x})\right).
\label{eq:MDN}
\end{equation}
In \eqref{eq:MDN}, $\{\pi_k(\myVec{x}), \myVec{\mu}_k(\myVec{x}), \myMat{\Sigma}_k(\myVec{x})\}_{k=1}^K$ denote the mixture weights, means, and covariance matrices. These  are obtained from $\myVec{x}$ using a dedicated  \ac{dnn}, trained separately from the \ac{cp} task using  $\mySet{D}_{\rm train}$.

The prediction region is obtained as a superlevel set of the learned conditional density, where the density threshold is determined through conformal calibration. The calibration set $\mySet{D}_{\mathrm{cal}}$ is used to construct the prediction region by evaluating the learned conditional density on the calibration samples and computing the nonconformity scores
\begin{equation}
R^{(i)} = -\log\hat{f}_{\hat{\myVec{\theta}}}\!\left(\myVec{y}^{(i)}\,\middle|\,\myVec{x}^{(i)}\right),
\label{eq:DRCPScore}
\end{equation}
where samples assigned higher probability by the predictive model yield smaller residuals. These residuals are then used to obtain the prediction region by setting $\hat Q^{1-\alpha}$ to the empirical $\left\lceil
\left(\left|\mySet{D}_{\mathrm{cal}}\right|+1\right)(1-\alpha)
\right\rceil$ quantile of $\{R^{(i)}\}$, resulting in
\begin{equation}
\mathcal C^{\alpha}_{\mathrm{DCP}}(\myVec{x})
= \left\{ \myVec{y}\in\mathcal Y: -\log\hat{f}_{\hat{\myVec{\theta}}}(\myVec{y}|\myVec{x}) \le \hat Q^{1-\alpha} \right\},
\label{eq:DRCPRegion}
\end{equation}
which contains all responses whose nonconformity score does not exceed the calibrated threshold.

\section{Conformalized Kalman Filter}
\label{sec: Method}
Interpreting the moments produced by a Kalman-type filter as defining a Gaussian posterior is generally unjustified under the nonlinear and non-Gaussian dynamics considered in~\ref{itm:nonlinear}, as well as under the model mismatch considered in~\ref{itm:mismatch}. Nonetheless, a straightforward application of existing \ac{cp} methods for multivariate regression is ill-suited to the considered tracking setting. Such methods are typically designed for a fixed regression task, whereas in dynamic systems the state distribution and the associated estimation-error distribution evolve over time. Applying them independently at each time instance would therefore either disregard this temporal variation or require separate prediction models and sufficient training and calibration data for each time index.

Our proposed approach combines the complementary strengths of Kalman-type filtering and \ac{cp}. Rather than treating the estimated first- and second-order moments, i.e., $(\hat{\myVec{s}}_t, \myMat{\Sigma}_t)$ in \eqref{eqn:EKFupdate}, as a complete probabilistic characterization of the state posterior, we interpret them as dynamically adapted statistical features that are informative of this posterior. Since these features are recursively updated from the \ac{ss} model and the acquired observations, they account for much of the temporal variation of the tracking problem and can consequently serve as inputs to \ac{cp} mechanisms. 

Based on this principle, Subsection~\ref{subsec: methods} introduces three conformalized Kalman filter constructions. These differ in the amount of data and learning they require, as well as in the geometry and flexibility of the confidence regions they support: an analytically initialized Gaussian construction, a learned quantile-based construction producing convex regions, and a learned distribution-based construction capable of representing multimodal posteriors.
The three constructions share a common conformal calibration and evaluation procedure, which transforms their initial confidence regions into statistically valid ones. In Subsection~\ref{subsec: Guarantees}, we establish finite-sample guarantees that apply to all considered constructions under exchangeability of the calibration and test trajectories. These results show that each method provides trustworthy confidence regions under both the sample-wise and trajectory-wise notions of coverage, although the resulting regions may differ substantially in their shape, tightness, and data requirements. We conclude the section with a comparative discussion in Subsection~\ref{subsec: Discussion}.

\subsection{Conformalized Kalman Filter Methods}
\label{subsec: methods}
 We next present three approaches for applying \ac{cp} on top of  Kalman-type filters.   The proposed methods differ in their nonconformity scores (which induce the geometry and flexibility of the confidence regions they construct), as well as in their data requirements. The first approach (\ac{cgkf}) directly exploits the moments produced by the filter and does not require learning an additional model, allowing all available labeled trajectories to be used for conformal calibration. The latter two approaches (\ac{cqkf} and \ac{cdkf}) follow the \ac{ml}-aided \ac{cp} methodologies reviewed in Subsection~\ref{ssec:Conformal Prediction}: they train a dedicated model to construct an initial confidence region from the filter outputs and therefore adopt the split \ac{cp} paradigm, partitioning the available data into disjoint training and calibration sets.  The main characteristics of the three approaches are summarized in Table~\ref{tab:cqkf_summary}.

\begin{table}
\centering
\caption{Summary of conformalized \ac{kf} methods}
\label{tab:cqkf_summary}
\renewcommand{\arraystretch}{1.2}
\begin{tabular}{|c|c|c|c|c|}
\hline
\textbf{Method} &
\textbf{Construction} &
\textbf{$\mathcal{D}_{\mathrm{train}}$} &
\textbf{$\mathcal{D}_{\mathrm{cal}}$} &
\textbf{Region Geometry} \\
\hline


\ac{cgkf} &
Analytical &
$\times$ &
\checkmark &
Gaussian \\
\hline

\ac{cqkf} &
Learned borders &
\checkmark &
\checkmark &
General convex  \\
\hline

\ac{cdkf} &
Learned distribution &
\checkmark &
\checkmark &
Gaussian mixture \\
\hline

\end{tabular}
\end{table}

\subsubsection{Common Calibration and Evaluation Procedure}
\label{subsubsec:common_operation}
All proposed conformalized Kalman filters share a common two-stage operation, illustrated in Fig.~\ref{fig:Trainin_Evaluation.pdf}. In an offline calibration stage, a selected Kalman-type filter is applied to the trajectories in $\mySet{D}_{\rm cal}$ to obtain the estimated posterior moments $\{(\hat{\myVec{s}}_t^{(i)},\myMat{\Sigma}_t^{(i)})\}$. Based on these moments, the ground-truth states, and the prescribed miscoverage level $\alpha$, each method computes a method-specific nonconformity score $R_t^{(i)}$ for every calibration trajectory $i$ and time instance $t$. These scores are then mapped into conformal correction terms according to the desired notion of coverage.

We set the conformal order-statistic index to
\begin{equation}
\label{eqn:conformal_order_statistic}
k = \left\lceil \bigl(|\mySet{D}_{\mathrm{cal}}|+1\bigr)(1-\alpha) \right\rceil.
\end{equation}
For sample-wise confidence as in \eqref{eqn:InstConfidence}, calibration is performed separately at each time instance. Specifically, for every $t\in\{1,\ldots,T\}$, we sort
$R_t^{(1)}\leq \cdots \leq R_t^{(|\mySet{D}_{\mathrm{cal}}|)}$
and set
\begin{equation}
\label{eqn:SmpWise}
\hat{Q}^{1-\alpha}(t)=R_t^{(k)}.
\end{equation}
The resulting time-dependent correction accounts for the evolution of the state and estimation-error distributions along the trajectory. For trajectory-wise confidence as in \eqref{eqn:TrajConfidence}, each calibration trajectory is instead assigned the aggregate score
\begin{equation}
\label{eqn:trajectory_score}
R^{(i)} = \max_{t\in\{1,\ldots,T\}}R_t^{(i)}. 
\end{equation}
After sorting
$R^{(1)}\leq\cdots\leq R^{(|\mySet{D}_{\mathrm{cal}}|)}$,
we set
\begin{equation}
\label{eqn:TrajWise}
\hat{Q}^{1-\alpha}(t) \equiv \hat{Q}^{1-\alpha} = R^{(k)}, 
\qquad t\in\{1,\ldots,T\}.
\end{equation}
Thus, a common correction is used  to control the probability that at least one state along the trajectory is not covered.

During evaluation, a Kalman-type filter processes the observed sequence and provides the moments $(\hat{\myVec{s}}_t,\myMat{\Sigma}_t)$. The selected conformalization method combines these moments with the previously computed correction $\hat{Q}^{1-\alpha}(t)$ to construct the confidence region $\mySet{C}_t^\alpha$. The specific definitions of the nonconformity score and the resulting region differ among the three proposed methods. A different Kalman-type implementation may be used at evaluation, provided that the underlying \ac{ss} model is the same as the one used in calibration.

\begin{figure*}
\centering
\includegraphics[width=0.95\textwidth]{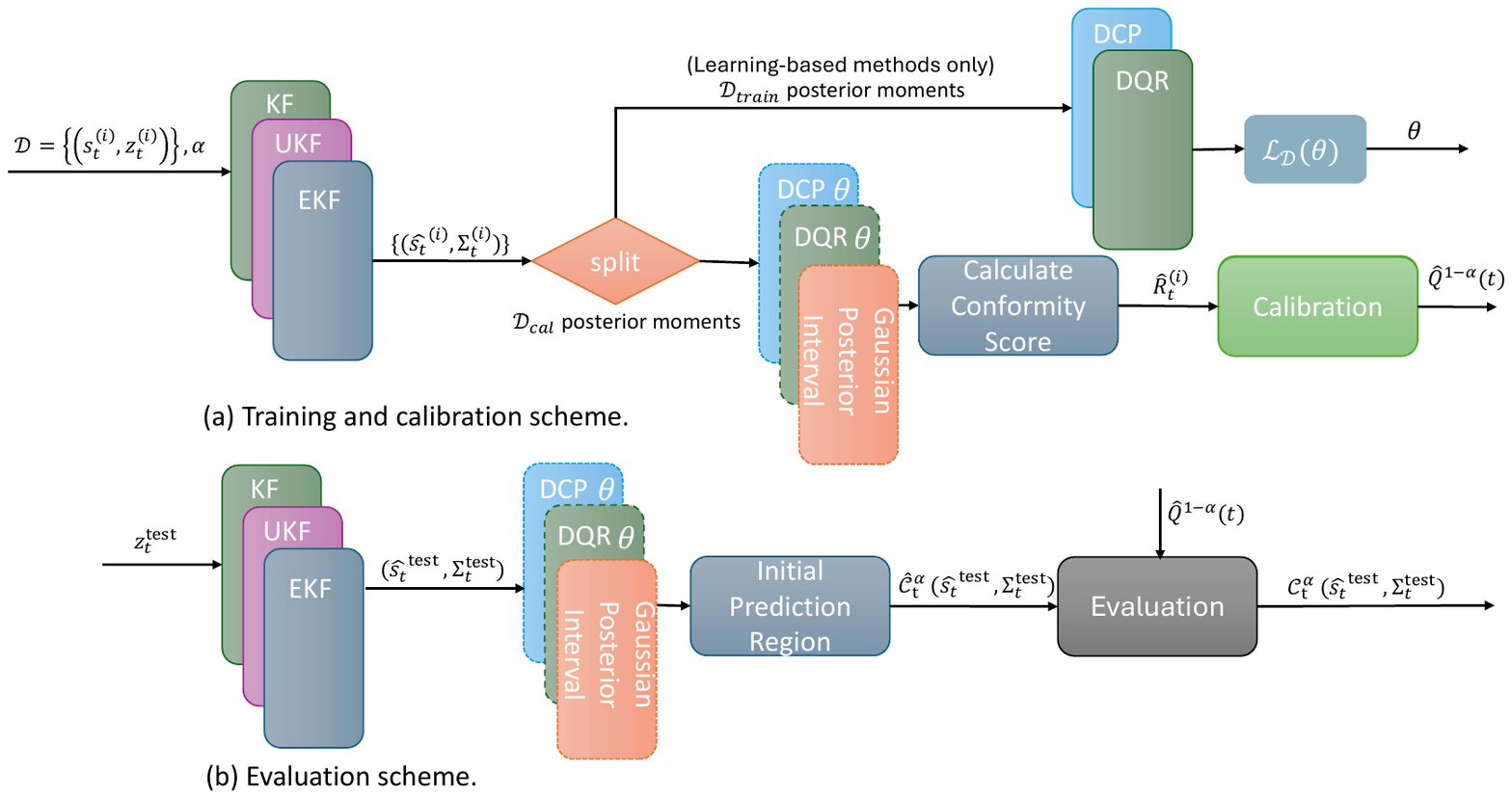}
\vspace{-1.6cm}
    \caption{High-level illustration of the proposed framework:
    $(a)$ Training and calibration scheme;
    $(b)$ Evaluation scheme.}
    \label{fig:Trainin_Evaluation.pdf}
\end{figure*}

 \subsubsection{\ac{cgkf}} 
\label{subsubsec:cgkf}
The simplest of the three proposed constructions uses the posterior moments $(\hat{\myVec{s}}_t,\myMat{\Sigma}_t)$ produced by the Kalman-type filter and does not require training an additional model. Specifically, \ac{cgkf} retains the ellipsoidal geometry induced by a Gaussian posterior, while using conformal calibration to resize the corresponding confidence region when this posterior approximation is inaccurate. Consequently, all available labeled trajectories can be allocated to calibration, i.e., $\mySet{D}_{\rm cal}$ is set to $\mySet{D}$ in \eqref{eqn:dataset}.  

\smallskip
{\bf Initial Prediction Region:}
Under the Gaussian approximation, the posterior distribution of the state is represented as
$\mathcal{N}(\hat{\myVec{s}}_t, \myMat{\Sigma}_t)$. Accordingly, the squared Mahalanobis distance 
$\|\myVec{s}_t-\hat{\myVec{s}}_t\|_{\myMat{\Sigma}_t}^2 \triangleq 
    (\myVec{s}_t-\hat{\myVec{s}}_t)^\top\myMat{\Sigma}_t^{-1}(\myVec{s}_t-\hat{\myVec{s}}_t)$
follows a chi-square distribution with $m$ degrees of freedom. This yields the initial confidence region
\begin{equation}
\label{eqn:ellipse_init}
\!\hat{\mySet{C}}_{t,\mathrm{G}}^{\alpha}(\hat{\myVec{s}}_t,\myMat{\Sigma}_t)\! =\! \left\{\myVec{s}\in\mathbb{R}^{m}  : \|\myVec{s}-\hat{\myVec{s}}_t\|_{\myMat{\Sigma}_t}^2 \! \leq F^{-1}_{\chi_m^2}(1\!\!-\alpha) \right\},
\end{equation}
where $F^{-1}_{\chi_m^2}$ denotes the inverse cumulative distribution function of a chi-square random variable with $m$ degrees of freedom. The resulting ellipsoid adapts its size, orientation, and eccentricity to the uncertainty encoded in $\myMat{\Sigma}_t$. When the posterior is indeed Gaussian, \eqref{eqn:ellipse_init} coincides with its $(1-\alpha)$ highest-density region. 

\smallskip
{\bf Nonconformity Score:} 
 We apply the underlying Kalman-type filter to each  trajectory in $\mySet{D}_{\rm cal}$ to obtain the corresponding posterior moments $\{(\hat{\myVec{s}}_t^{(i)},\myMat{\Sigma}_t^{(i)})\}$.
The nonconformity score is 
\begin{equation}
\label{eqn:residuals_Ellipse}
R_t^{(i)}  = \big\|\myVec{s}_t^{(i)}-\hat{\myVec{s}}_t^{(i)}\big\|_{\myMat{\Sigma}_t^{(i)}}^2  
- F^{-1}_{\chi_m^2}(1-\alpha).
\end{equation}
Thus, $R_t^{(i)}$ measures the signed deviation of the true state from the boundary of the initial Gaussian region: positive values correspond to states lying outside the region, while negative values indicate states lying inside it.

\smallskip
{\bf  Evaluation:}
At evaluation time, the Kalman-type filter processes the newly observed sequence and produces $(\hat{\myVec{s}}_t,\myMat{\Sigma}_t)$. The  \ac{cgkf} region is obtained by adjusting the Gaussian Mahalanobis threshold according to the calibration quantile, i.e., 
\begin{align}
&\mySet{C}_{t,{\mathrm{G}}}^\alpha (\hat{\myVec{s}}_t,\myMat{\Sigma}_t)\notag \\ 
&= \left\{ \myVec{s}\in\mathbb{R}^{m} :  \|\myVec{s}-\hat{\myVec{s}}_t\|_{\myMat{\Sigma}_t}^2\! \leq F^{-1}_{\chi_m^2}(1\!-\!\alpha) \!+\! \hat{Q}^{1-\alpha}(t) \right\}.
\label{eqn:ellipse_cp_region}
\end{align}
 A positive conformal correction expands the nominal Gaussian region, whereas a negative correction contracts it. In both cases, the shape and orientation of the region remain determined by the covariance propagated by the Kalman-type filter, while its effective radius is calibrated from data.

\subsubsection{\ac{cqkf}}
\label{subsubsec:cqkf}
The second method replaces the Gaussian posterior approximation used by \ac{cgkf} with a learned quantile-based characterization of the state uncertainty. Specifically, the posterior moments $(\hat{\myVec{s}}_t,\myMat{\Sigma}_t)$ produced by the Kalman-type filter are treated as input features to a \ac{dqr} model~\cite{feldman2023calibrated}. This allows the resulting confidence regions to adapt to asymmetric and non-Gaussian posterior distributions without imposing a parametric distributional model. Unlike \ac{cgkf}, \ac{cqkf} requires the available data \eqref{eqn:dataset} to be divided into disjoint training and calibration sets, denoted $\mySet{D}_{\rm train}$ and $\mySet{D}_{\rm cal}$, respectively.  

\smallskip
{\bf Initial Prediction Region:}
We apply the selected Kalman-type filter to each trajectory in $\mySet{D}_{\rm train}$, obtaining the posterior moments
$\{(\hat{\myVec{s}}_t^{(i)},\myMat{\Sigma}_t^{(i)})\}$.
Given a finite set of unit-norm projection directions
$\mySet{U}\subset\mathbb{S}^{m-1}$, a directional quantile regression \ac{dnn}
$\hat{\mu}_{\myVec{\theta}}$ (which in our numerical study is set to a \ac{mlp} with two hidden layers) is trained to map the posterior moments to a vector of directional quantiles, where each output dimension corresponds to a fixed direction $\myVec{u}\in\mySet{U}$
Following the formulation in \eqref{eqn:cr_leaning}, its parameters $\hat{\myVec{\theta}}$ are trained based on the empirical risk
\begin{align*}
\arg\min_{\myVec{\theta}} \sum_{i\in\mySet{D}_{\rm train}} \sum_{t=1}^{T} \sum_{\myVec{u}\in\mySet{U}} \rho_{\alpha} \left( \myVec{u}^{\top}\myVec{s}_t^{(i)}, \hat{\mu}_{\myVec{\theta}} \big( (\hat{\myVec{s}}_t^{(i)},\myMat{\Sigma}_t^{(i)}), \myVec{u} \big)\right), 
\end{align*}
where $\rho_{\alpha}$ is defined in \eqref{eqn:Pinball}. For a new pair of posterior moments, the learned directional quantiles define the initial  region
\begin{align}
\!\!\hat{\mySet{C}}_{t,\mathrm{Q}}^{\alpha} (\hat{\myVec{s}}_t,\myMat{\Sigma}_t) \!=\!
\bigcap_{\myVec{u}\in\mySet{U}} \!\left\{ \myVec{s}\in\mathbb{R}^{m}: \myVec{u}^{\top}\myVec{s} \!\geq\! \hat{\mu}_{\hat{\myVec{\theta}}} \big( (\hat{\myVec{s}}_t,\myMat{\Sigma}_t), \myVec{u} \big) \right\}.
\label{eqn:cqkf_initial_region}
\end{align}
As an intersection of directional half-spaces, \eqref{eqn:cqkf_initial_region} is convex. 
Its size and geometry vary with the moments produced by the Kalman-type filter and, through data-driven training, can capture asymmetric uncertainty not represented by the Gaussian ellipsoid in \eqref{eqn:ellipse_init}.

\smallskip
{\bf Nonconformity Score:}
The Kalman-type filter is applied to each calibration trajectory in $\mySet{D}_{\rm cal}$ to obtain $\{(\hat{\myVec{s}}_t^{(i)},\myMat{\Sigma}_t^{(i)})\}$.
For each state and time instance, the conformity score is defined as the largest violation of the learned directional constraints:
\begin{equation}
\label{eq:Residual_dqr_kalman}
R_t^{(i)}  = \max_{\myVec{u}\in\mySet{U}} \left\{ \hat{\mu}_{\hat{\myVec{\theta}}} \big( (\hat{\myVec{s}}_t^{(i)},\myMat{\Sigma}_t^{(i)}), \myVec{u} \big) - \myVec{u}^{\top}\myVec{s}_t^{(i)} \right\}.
\end{equation}
 The scores in \eqref{eq:Residual_dqr_kalman} are mapped into the correction terms $\hat{Q}^{1-\alpha}(t)$ using the sample-wise or trajectory-wise calibration procedure.

\smallskip
{\bf Evaluation:}
At evaluation time, the trained quantile model receives the posterior moments $(\hat{\myVec{s}}_t,\myMat{\Sigma}_t)$ produced by the Kalman-type filter. The conformal correction relaxes each directional constraint by $\hat{Q}^{1-\alpha}(t)$, yielding
\begin{align}
&\mySet{C}_{t,\mathrm{Q}}^{\alpha} (\hat{\myVec{s}}_t,\myMat{\Sigma}_t) \notag \\
&=
\bigcap_{\myVec{u}\in\mySet{U}}\! \left\{ \myVec{s}\in\mathbb{R}^{m}: \myVec{u}^{\top}\myVec{s} \!\geq\! \hat{\mu}_{\hat{\myVec{\theta}}} \big( (\hat{\myVec{s}}_t,\myMat{\Sigma}_t), \myVec{u} \big)\! -\! \hat{Q}^{1-\alpha}(t) \right\}.
\label{eqn:cqkf_cp_region}
\end{align}
A positive correction enlarges the initial region by relaxing all its half-space constraints, whereas a negative correction contracts it. The resulting confidence region remains convex, while its boundaries adapt to the state- and time-dependent uncertainty encoded in the moments produced by the Kalman-type filter.

\subsubsection{\ac{cdkf}}
\label{subsubsec:cdkf}

\ac{cdkf} learns a conditional representation of the state posterior from the moments produced by the Kalman-type filter. Specifically, $(\hat{\myVec{s}}_t,\myMat{\Sigma}_t)$ are used as input features to a conditional density estimator
$\hat{f}_{\myVec{\theta}}(\myVec{s}_t\mid \hat{\myVec{s}}_t,\myMat{\Sigma}_t)$,
with the corresponding state $\myVec{s}_t$ serving as the target. In contrast to \ac{cqkf}, which is restricted to convex confidence regions, the learned density can represent multimodal posterior distributions and consequently induce non-convex, and possibly disconnected, confidence regions. As in \ac{cqkf}, the available data are divided into disjoint training and calibration sets, denoted $\mySet{D}_{\rm train}$ and $\mySet{D}_{\rm cal}$, respectively.

\smallskip
{\bf Initial Prediction Region:}
We apply the selected Kalman-type filter to each trajectory in $\mySet{D}_{\rm train}$ to obtain the posterior moments
$\{(\hat{\myVec{s}}_t^{(i)},\myMat{\Sigma}_t^{(i)})\}$.
A dedicated \ac{dnn} with parameters $\myVec{\theta}$ maps these moments into the parameters of a $K$-component Gaussian mixture model (specifically, in our numerical study we used a mixture density network comprising a three-layered \ac{mlp}). In particular, the resulting conditional density is expressed as
\begin{align*}
&\hat{f}_{\myVec{\theta}} \left( \myVec{s}\mid \hat{\myVec{s}}_t,\myMat{\Sigma}_t\right)
\notag\\
&\quad =\sum_{k=1}^{K}\pi_{k,\myVec{\theta}}\left(\hat{\myVec{s}}_t,\myMat{\Sigma}_t\right) \mathcal{N}\left(\myVec{s};\myVec{\mu}_{k,\myVec{\theta}}\left(\hat{\myVec{s}}_t,\myMat{\Sigma}_t\right),\myMat{\Gamma}_{k,\myVec{\theta}} \left(\hat{\myVec{s}}_t,\myMat{\Sigma}_t\right)\right), 
\end{align*}
where $\pi_{k,\myVec{\theta}}$, $\myVec{\mu}_{k,\myVec{\theta}}$, and $\myMat{\Gamma}_{k,\myVec{\theta}}$
denote the mixture weight, mean, and covariance of the $k$th component, respectively, taken from the output of the \ac{dnn}, with $\pi_{k,\myVec{\theta}}\geq 0$ and $\sum_{k=1}^{K}\pi_{k,\myVec{\theta}}=1$.

The network parameters are trained by maximizing the empirical  log-likelihood over $\mySet{D}_{\rm train}$:
\begin{align*}
\arg\max_{\myVec{\theta}} \sum_{i\in\mySet{D}_{\rm train}} \sum_{t=1}^{T}\log \hat{f}_{\myVec{\theta}} \left( \myVec{s}_t^{(i)} \,\middle|\, \hat{\myVec{s}}_t^{(i)}, \myMat{\Sigma}_t^{(i)} \right). 
\end{align*}
The learned  parameters $\hat{\myVec{\theta}}$ yield a density that  defines a nested family of candidate prediction regions, parameterized by a  threshold $q$, as
\begin{align}
\hat{\mySet{C}}_{t,\mathrm{D}} \left( q; \hat{\myVec{s}}_t,\myMat{\Sigma}_t \right)
 =
\left\{ \myVec{s}\in\mathbb{R}^{m}: \log \hat{f}_{\hat{\myVec{\theta}}} \left( \myVec{s} \,\middle|\, \hat{\myVec{s}}_t,\myMat{\Sigma}_t \right) \geq q \right\}.
\label{eqn:cdkf_initial_region}
\end{align}
 Owing to the Gaussian-mixture representation, its geometry can adapt to asymmetric and multimodal posterior distributions.

\smallskip
{\bf Nonconformity Score:}
The Kalman-type filter is applied to each calibration trajectory in $\mySet{D}_{\rm cal}$ to obtain
$\{(\hat{\myVec{s}}_t^{(i)},\myMat{\Sigma}_t^{(i)})\}$.
Following the \ac{dcp} formulation, the nonconformity score is the negative log-likelihood assigned to the true state:
\begin{equation}
\label{eq:residual_drcp_kalman}
R_t^{(i)} = -\log \hat{f}_{\hat{\myVec{\theta}}} \left( \myVec{s}_t^{(i)} \,\middle|\, \hat{\myVec{s}}_t^{(i)}, \myMat{\Sigma}_t^{(i)} \right).
\end{equation}
Thus, states assigned high probability by the learned  posterior yield small nonconformity scores. The scores in \eqref{eq:residual_drcp_kalman} are mapped into $\hat{Q}^{1-\alpha}(t)$ using the common calibration procedure.

\smallskip
{\bf Evaluation:}
At evaluation time, the Kalman-type filter produces the moments
$(\hat{\myVec{s}}_t,\myMat{\Sigma}_t)$, which are mapped by the trained \ac{dnn} into the parameters of the Gaussian-mixture posterior approximation. The calibrated \ac{cdkf} region is then given by
\begin{align}
\!\!\!\mySet{C}_{t,\mathrm{D}}^{\alpha}\! \left( \hat{\myVec{s}}_t,\myMat{\Sigma}_t \right)
  \!=\!
\left\{ \myVec{s}\in\mathbb{R}^{m}: \log \hat{f}_{\hat{\myVec{\theta}}} \left( \myVec{s} \,\middle|\, \hat{\myVec{s}}_t,\myMat{\Sigma}_t \right)\! \geq\! -\hat{Q}^{1\!-\!\alpha}(t) \right\}.
\label{eqn:cdkf_cp_region}
\end{align}
The resulting region adapts to the posterior representation inferred from the Kalman moments and may consist of multiple disconnected components, allowing it to characterize uncertainty patterns that cannot be represented by the ellipsoidal or convex regions of \ac{cgkf} and \ac{cqkf}, respectively.

\subsection{Theoretical Guarantees}
\label{subsec: Guarantees}
Although \ac{cgkf}, \ac{cqkf}, and \ac{cdkf} differ in their construction, data requirements, and supported confidence-region geometries, they share the same statistical validity guarantees. These guarantees follow from the common calibration procedure  and do not depend on  the method-specific nonconformity score. 

\subsubsection{Exchangeability}
We first formalize the exchangeability condition underlying our coverage results.

\begin{definition}[Exchangeability]
\label{def:exch}
Let $\myVec{X}^{(i)} \triangleq \big\{\bigl(\myVec{s}_t^{(i)},\myVec{z}_t^{(i)}\bigr) \big\}_{t=1}^{T}$,  $i=1,\ldots,N$, denote a collection of trajectories. The trajectories $\{\myVec{X}^{(i)}\}_{i=1}^{N}$ are said to be \emph{exchangeable} if their joint distribution is invariant under permutations of the trajectory indices. Namely, for every permutation $\pi$ of $\{1,\ldots,N\}$,
\begin{equation}
\left( \myVec{X}^{(1)},\ldots,\myVec{X}^{(N)} \right) \overset{d}{=} \left( \myVec{X}^{(\pi(1))},\ldots,\myVec{X}^{(\pi(N))} \right),
\label{eqn:exchangeability}
\end{equation}
where $\overset{d}{=}$ denotes equality in distribution.
\end{definition}

Exchangeability is weaker than i.i.d. sampling, as every collection of i.i.d. trajectories is exchangeable, while exchangeable trajectories need not be mutually independent.  Definition~\ref{def:exch} concerns exchangeability \emph{across trajectories}; it does not impose independence or stationarity across the time indices within each trajectory, which may exhibit  temporal dependencies induced by the underlying \ac{ss} model. For the learning-based \ac{cqkf} and \ac{cdkf} methods, the following results are understood conditionally on the disjoint training set $\mySet{D}_{\rm train}$ and on the resulting learned parameters. Thus, the Kalman-type filtering and conformalization mechanisms are fixed before processing the calibration and test trajectories.

\subsubsection{Sample-Wise Calibration}
We first consider the sample-wise calibration procedure in \eqref{eqn:SmpWise}, for which we obtain the following coverage guarantee:

\begin{theorem}[Sample-wise coverage]
\label{thm:c-kf-samplewise}
Fix one of the proposed conformalization methods and a Kalman-type filter. Let
$ \mySet{D}_{\rm cal} \cup     \left\{ \bigl(\myVec{s}_t,\myVec{z}_t\bigr) \right\}_{t=1}^{T}$,
be exchangeable calibration and test trajectories generated by the \ac{ss} model \eqref{eq:ssModel}. For a prescribed miscoverage level $\alpha\in(0,1)$, let the correction terms $\hat{Q}^{1-\alpha}(t)$ be computed according to \eqref{eqn:SmpWise}. Then, the resulting confidence regions satisfy
\begin{equation}
\Pr\left( \myVec{s}_t \in \mySet{C}_t^{\alpha} \bigl(\hat{\myVec{s}}_t,\myMat{\Sigma}_t\bigr) \right) \geq 1-\alpha, \qquad
\forall t\in\{1,\ldots,T\}.
\label{eqn:c-kf-samplewise-coverage}
\end{equation}
Moreover, if the calibration and test nonconformity scores at time $t$ are almost surely distinct, then $\forall t\in\{1,\ldots,T\}$ it holds that
\begin{align}
\Pr\left( \myVec{s}_t \in \mySet{C}_t^{\alpha} \bigl(\hat{\myVec{s}}_t,\myMat{\Sigma}_t\bigr) \right) &\leq 1-\alpha+ \frac{1}{|\mySet{D}_{\rm cal}|+1}.
\label{eqn:c-kf-samplewise-upper}
\end{align}
\end{theorem}

\ifproofs
\begin{IEEEproof}
    The proof is provided in Appendix~\ref{app:proof1}.
\end{IEEEproof}
\fi

\smallskip
Theorem~\ref{thm:c-kf-samplewise} follows from the exchangeability of the scores at each fixed time instance. It indicates that all  proposed constructions achieve the prescribed marginal coverage, regardless of whether their initial regions are ellipsoidal, convex, or density-based. Furthermore, \eqref{eqn:c-kf-samplewise-upper} shows that their finite-sample coverage exceeds the target level by at most
$1/(|\mySet{D}_{\rm cal}|+1)$.

\subsubsection{Trajectory-Wise Calibration}
A mechanism calibrated for sample-wise coverage can also provide a trajectory-wise guarantee by tightening its prescribed miscoverage level, as stated in the following corollary:

\begin{corollary}[Trajectory-wise coverage using sample-wise quantiles]
\label{cor:c-kf-bonferroni}
Consider the setting of Theorem~\ref{thm:c-kf-samplewise}, and construct the sample-wise regions with $\tilde{\alpha} = \frac{\alpha}{T}$.  
Then,
\begin{equation}
\Pr\left( \myVec{s}_t \in \mySet{C}_t^{\tilde{\alpha}} \bigl(\hat{\myVec{s}}_t,\myMat{\Sigma}_t\bigr), \ \forall t\in\{1,\ldots,T\} \right)
\geq 1-\alpha.
\label{eqn:c-kf-bonferroni-coverage}
\end{equation}
\end{corollary}

\ifproofs
\begin{IEEEproof}
The proof is provided in Appendix~\ref{app:proof2}.
\end{IEEEproof}
\fi

\smallskip
Corollary~\ref{cor:c-kf-bonferroni} follows by applying Theorem~\ref{thm:c-kf-samplewise} with level $\tilde{\alpha}$ and using the union bound over the $T$ time instances. This approach corresponds to a {\em Bonferroni-type} correction~\cite{bonferroni1936teoria}. It provides trajectory-wise confidence using the sample-wise calibration mechanism, but may be conservative since every individual region is constructed with the stricter confidence level $1-\alpha/T$.

Alternatively, trajectory-wise confidence can be obtained directly using the aggregated scores in \eqref{eqn:trajectory_score} and their corresponding quantile in \eqref{eqn:TrajWise}, as stated next:

\begin{theorem}[Trajectory-wise coverage]
\label{thm:c-kf-trajectorywise}
Under the same setting as in Theorem~\ref{thm:c-kf-samplewise}, let the conformal correction be computed using the trajectory-wise calibration procedure in \eqref{eqn:trajectory_score}--\eqref{eqn:TrajWise}. Then, the resulting confidence regions satisfy
\begin{equation}
\Pr\left( \myVec{s}_t \in \mySet{C}_t^{\alpha} \bigl(\hat{\myVec{s}}_t,\myMat{\Sigma}_t\bigr), \ \forall t\in\{1,\ldots,T\} \right)
\geq 1-\alpha.
\label{eqn:c-kf-trajectorywise-coverage}
\end{equation}
\end{theorem}

\ifproofs
\begin{IEEEproof}
The proof is provided in Appendix~\ref{app:proof3}.
\end{IEEEproof}
\fi

\smallskip
Theorem~\ref{thm:c-kf-trajectorywise} follows by treating the maximum nonconformity score of each trajectory as a single exchangeable score. It thus controls the probability that any state along the test trajectory lies outside its corresponding confidence region. Unlike Corollary~\ref{cor:c-kf-bonferroni}, it uses a single $(1-\alpha)$ quantile of the trajectory-level scores and does not separately tighten the confidence level at every time instance. It also has a substantially milder calibration-data requirement: selecting a finite $(1-\alpha)$ conformal quantile requires
$|\mySet{D}_{\rm cal}| \geq \frac{1-\alpha}{\alpha}$, 
whereas the Bonferroni construction at level $\tilde{\alpha}=\alpha/T$ requires
$|\mySet{D}_{\rm cal}| \geq \frac{1-\alpha/T}{\alpha/T} =
\frac{T}{\alpha}-1$.
The direct trajectory-wise construction is therefore generally preferable when joint coverage over a complete horizon is required. Nevertheless, neither validity result guarantees that the resulting regions are tight. Their geometry and efficiency depend on the  uncertainty representation and the nonconformity scores.

\subsection{Discussion}
\label{subsec: Discussion} 

Our conformalized Kalman filtering framework  addresses the challenges identified in Subsection~\ref{sec:Problem Formulation}. In particular, the confidence regions are not required to coincide with the Gaussian posterior implied by the first- and second-order moments, allowing their coverage to remain valid under nonlinear and non-Gaussian dynamics as in~\ref{itm:nonlinear}, as well as under model mismatch as in~\ref{itm:mismatch}. The proposed methods retain the recursive statistical information produced by the underlying Kalman-type filter: the moments $(\hat{\myVec{s}}_t,\myMat{\Sigma}_t)$ adapt with the evolving dynamics and observations, and thus provide time-dependent features from which the confidence regions are constructed. This also makes the framework largely agnostic to the particular Kalman-type algorithm, addressing~\ref{itm:Alg}, as the conformalization operates on its posterior moment estimates rather than on its internal recursion. 

The three proposed constructions provide different levels of flexibility within this framework: \ac{cgkf} is most natural when the posterior geometry is reasonably captured by a Gaussian approximation, \ac{cqkf} accommodates asymmetric and non-Gaussian uncertainty through learned convex regions, while \ac{cdkf} further supports multimodal distributions and non-convex confidence regions. The validity guarantees of Subsection~\ref{subsec: Guarantees} apply to all three constructions irrespective of these modeling choices; their main effect is therefore on the tightness and geometry of the resulting regions.

{\bf Complexity Analysis:} 
The increased flexibility of our constructions comes with a  computational and data-complexity tradeoff. All methods incur the cost of the underlying Kalman-type filter, while \ac{cgkf} adds only the computation of a Mahalanobis-type score and its calibrated threshold, and requires no additional training; it is thus the least computationally demanding and allows all available labeled trajectories to be used for calibration.  \ac{cqkf} and \ac{cdkf} require an additional offline learning stage and divide the available data between training and calibration. For \ac{cqkf}, the additional inference cost is dominated by evaluating the directional quantile \ac{dnn} and forming the constraints associated with the directions in $\mySet{U}$; hence, its complexity and the resolution of the resulting convex region both increase with $|\mySet{U}|$. For \ac{cdkf}, the \ac{dnn} predicts the parameters of a $K$-component Gaussian mixture, and evaluation requires computing the  mixture density, with complexity increasing with both $K$ and the state dimension through the  covariance matrices. Thus, \ac{cgkf} offers the lowest-complexity construction at the expense of restricted geometry, while the learning-based \ac{cqkf} and \ac{cdkf} trade additional training and inference complexity for more expressive confidence regions. These tradeoffs are quantified numerically in Section~\ref{sec: Numerical Study}. 

{\bf Extensions:}
Several extensions of our framework are left for future investigation. First, the guarantees established in Subsection~\ref{subsec: Guarantees} rely on exchangeability between the calibration and test trajectories. While this assumption permits arbitrary temporal dependence within each trajectory, it may be violated when the underlying operating conditions change. Extending conformalized Kalman filtering to distribution shifts or other non-exchangeable settings can potentially be obtained by incorporating weighted, localized, or adaptive conformal mechanisms without assuming access to the true latent state~\cite{barber2023conformal}. Second, our calibration procedure is defined over the time horizon $T$ represented in the available trajectories. The sample-wise corrections $\hat{Q}^{1-\alpha}(t)$ are characterized for $t\in\{1,\ldots,T\}$, while the trajectory-wise guarantee also pertains to this prescribed horizon. Providing trustworthy regions when filtering continues beyond the horizon represented in the calibration data requires mechanisms that can extrapolate or otherwise parameterize the temporal evolution of the conformal correction. Finally, the underlying idea of exploiting model-based posterior information as features for  calibration is not intrinsically limited to Kalman-type filters. Extending this principle to more general Bayesian filtering methods could enable conformal regions that exploit  posterior representations beyond first- and second-order moments.

\section{Numerical Study}
\label{sec: Numerical Study}

We numerically evaluate the proposed conformalized Kalman filtering framework in terms of its coverage reliability, statistical efficiency, and computational complexity\footnote{
The source code and the hyperparameters can be found online  at: \url{https://github.com/OlgaWeisman/CxKF}}. We consider scalar and multivariate \ac{ss} models spanning linear Gaussian, non-Gaussian, nonlinear, mismatched, and chaotic dynamics. The study is organized as follows. We first describe the considered models, benchmarks, and evaluation protocol. We then report the scalar-state and multivariate-state results, followed by a study of the computational complexity and the effect of the number of directional quantiles used by \ac{cqkf}.

\subsection{Experimental Setup}
\label{subsec:experimental_setup}

\noindent
{\bf State-Space Models and Filters:}
For all considered models, we denote by $q^2$ and $r^2$ the process- and measurement-noise levels associated with $\myMat{Q}_t$ and $\myMat{R}_t$, respectively. We define
\begin{equation}
    \mathrm{SNR}= 10\log_{10}\left(\frac{1}{r^2}\right)\;[\mathrm{dB}],
    \qquad
    \frac{q^2}{r^2}=0.01,
\label{eqn:numerical_snr}
\end{equation}
such that the process-to-measurement noise ratio remains fixed as the \ac{snr} varies. For linear systems we employ the \ac{kf}, and for  nonlinear ones we use the \ac{ekf} and \ac{ukf}, allowing us to evaluate the compatibility of the proposed approach with different Kalman-type filters as required by~\ref{itm:Alg}. 

For the scalar experiments, we consider three settings. The first is the linear Gaussian model $f_t(s)=0.9s$ and $h_t(s)=s$. The second uses the same linear dynamics with Laplacian rather than Gaussian noise, thus violating the Gaussianity assumption. The third is a nonlinear setting with $f_t(s)=\sin(s)$ and $h_t(s)=s^2$. In this case, we additionally introduce model mismatch by providing the filter with the approximation $\tilde f_t(s)=s$ while generating trajectories according to the true nonlinear transition $f_t(s)=\sin(s)$.

For the multivariate experiments, we consider the physically motivated nonlinear pendulum model from~\cite{dahan2024uncertainty}, whose state is
$\myVec{s}_t=[\theta_t,\dot{\theta}_t]^\top$ and whose transition and
observation functions are
\begin{equation}
\myVec{f}(\myVec{s}_t)=
\begin{bmatrix}
\theta_t+\dot{\theta}_t\Delta t\\
\dot{\theta}_t-\frac{g}{\ell}\sin(\theta_t)\Delta t
\end{bmatrix},
\quad
\myVec{h}(\myVec{s}_t)=
\begin{bmatrix}
\ell\cos(\theta_t)\\
\ell\sin(\theta_t)
\end{bmatrix},
\label{eqn:pendulum_numeric}
\end{equation}
with $\ell=1$ and $\Delta t=0.02$. 


Finally, we consider the three-dimensional Lorenz attractor as a challenging chaotic tracking scenario. Its continuous-time dynamics are represented as
\begin{equation}
    \frac{d\myVec{s}_{\tau}}{d\tau} = \myMat{A}(\myVec{s}_{\tau})\myVec{s}_{\tau},
    \quad
    \myMat{A}(\myVec{s}_{\tau})  =
    \begin{bmatrix}
        -10 & 10 & 0\\
        28 & -1 & -s_{1,\tau}\\
        0 & s_{1,\tau} & -8/3
    \end{bmatrix}.
\label{eqn:lorenz_numeric}
\end{equation}
As in \cite{revach2022kalmannet}, we discretize the dynamics with $\Delta\tau=0.02$ using $\myVec{s}_{t+1}=\myMat{F}(\myVec{s}_t)\myVec{s}_t$, where $\myMat{F}(\myVec{s}_t)=\exp(\myMat{A}(\myVec{s}_t)\Delta\tau)$ is approximated by a fifth-order Taylor expansion.

\newcommand{\SelectedSNR}{0}


\newcommand{\GetMetric}[4]{%
    \pgfplotstableread[col sep=comma]{#1}\datatable%
    \def\foundvalue{}%
    \pgfplotstableforeachcolumnelement{algo}\of\datatable\as\algovalue{%
        \pgfplotstablegetelem{\pgfplotstablerow}{snr_db}\of\datatable%
        \edef\snrvalue{\pgfplotsretval}%
        \IfStrEq{\algovalue}{#2}{%
            \pgfmathparse{abs(\snrvalue-(#3)) < 0.001 ? 1 : 0}%
            \ifnum\pgfmathresult=1
                \pgfplotstablegetelem{\pgfplotstablerow}{#4}\of\datatable%
                \xdef\foundvalue{\pgfplotsretval}%
            \fi
        }{}%
    }%
    \ifx\foundvalue\empty
        --%
    \else
        \pgfmathprintnumber[
            fixed,
            fixed zerofill,
            precision=2
        ]{\foundvalue}%
    \fi
}


\newcommand{\GetMetricThree}[4]{%
    \pgfplotstableread[col sep=comma]{#1}\datatable%
    \def\foundvalue{}%
    \pgfplotstableforeachcolumnelement{algo}\of\datatable\as\algovalue{%
        \pgfplotstablegetelem{\pgfplotstablerow}{snr_db}\of\datatable%
        \edef\snrvalue{\pgfplotsretval}%
        \IfStrEq{\algovalue}{#2}{%
            \pgfmathparse{abs(\snrvalue-(#3)) < 0.001 ? 1 : 0}%
            \ifnum\pgfmathresult=1
                \pgfplotstablegetelem{\pgfplotstablerow}{#4}\of\datatable%
                \xdef\foundvalue{\pgfplotsretval}%
            \fi
        }{}%
    }%
    \ifx\foundvalue\empty
        --%
    \else
        \pgfmathprintnumber[
            fixed,
            fixed zerofill,
            precision=2
        ]{\foundvalue}%
    \fi
}


\newcommand{\GetMetricThreeBold}[4]{%
    \pgfplotstableread[col sep=comma]{#1}\datatable%
    \def\foundvalue{}%
    \pgfplotstableforeachcolumnelement{algo}\of\datatable\as\algovalue{%
        \pgfplotstablegetelem{\pgfplotstablerow}{snr_db}\of\datatable%
        \edef\snrvalue{\pgfplotsretval}%
        \IfStrEq{\algovalue}{#2}{%
            \pgfmathparse{abs(\snrvalue-(#3)) < 0.001 ? 1 : 0}%
            \ifnum\pgfmathresult=1
                \pgfplotstablegetelem{\pgfplotstablerow}{#4}\of\datatable%
                \xdef\foundvalue{\pgfplotsretval}%
            \fi
        }{}%
    }%
    \ifx\foundvalue\empty
        --%
    \else
        {\bfseries\boldmath
        \pgfmathprintnumber[
            fixed,
            fixed zerofill,
            precision=2
        ]{\foundvalue}}%
    \fi
}


\newcommand{\PS}[3]{%
    \GetMetric{#1}{#2}{#3}{C_mean}%
}

\newcommand{\PT}[3]{%
    \GetMetric{#1}{#2}{#3}{S}%
}

\newcommand{\GetRawMetric}[5]{%
    \pgfplotstableread[col sep=comma]{#1}\datatable%
    \xdef#5{}%
    \pgfplotstableforeachcolumnelement{algo}\of\datatable\as\algovalue{%
        \pgfplotstablegetelem{\pgfplotstablerow}{snr_db}\of\datatable%
        \edef\snrvalue{\pgfplotsretval}%
        \IfStrEq{\algovalue}{#2}{%
            \pgfmathparse{abs(\snrvalue-(#3)) < 0.001 ? 1 : 0}%
            \ifnum\pgfmathresult=1
                \pgfplotstablegetelem{\pgfplotstablerow}{#4}\of\datatable%
                \xdef#5{\pgfplotsretval}%
            \fi
        }{}%
    }%
}

\newcommand{\GetNormalizedMetric}[5]{%
    \GetRawMetric{#1}{#2}{#3}{#4}{\methodvalue}%
    \GetRawMetric{#1}{#5}{#3}{WI_mean}{\referencevalue}%
    \ifx\methodvalue\empty
        --%
    \else
        \pgfmathparse{\methodvalue/\referencevalue}%
        \pgfmathprintnumber[fixed,fixed zerofill,precision=2]{\pgfmathresult}%
    \fi
}

\newcommand{\RS}[3]{%
    \def\referencealgo{KF-Gauss}%
    \IfSubStr{#2}{Tj}{\def\referencealgo{KF-Gauss-Bonf}}{}%
    \IfSubStr{#2}{TS}{\def\referencealgo{KF-Gauss-Bonf}}{}%
    \IfSubStr{#2}{Bonf}{\def\referencealgo{KF-Gauss-Bonf}}{}%
    \GetNormalizedMetric{#1}{#2}{#3}{WI_mean}{\referencealgo}%
    /%
    \GetNormalizedMetric{#1}{#2}{#3}{WI_svd}{\referencealgo}%
}

\newcommand{\RSB}[3]{%
    {\bfseries\boldmath
    \RS{#1}{#2}{#3}%
    }%
}


\begin{table*}[t]
\centering

\setlength{\tabcolsep}{5pt}
\renewcommand{\arraystretch}{1.05}
\footnotesize

\begin{adjustbox}{max width=\textwidth}

\begin{tabular}{@{}llcccc@{}}
\toprule

\multirow{2}{*}{\textbf{Scenario}}
&
\multirow{2}{*}{\textbf{Method}}
&
\multicolumn{2}{c}{\textbf{Sample-wise}}
&
\multicolumn{2}{c}{\textbf{Trajectory-wise}}
\\

\cmidrule(lr){3-4}
\cmidrule(lr){5-6}

&
&
\textbf{Per-sample Miscoverage [\%]}
&
\textbf{Normalized Interval Width}
&
\textbf{Per-trajectory Miscoverage [\%]}
&
\textbf{Normalized Interval Width}
\\

\midrule


\multirow{8}{*}{\textbf{Lin. Gauss.}}

& Gauss
& \PS{csvs/metrics_LC-Gauss.csv}{KF-Gauss}{\SelectedSNR}
& \RSB{csvs/metrics_LC-Gauss.csv}{KF-Gauss}{\SelectedSNR}
& --
& --
\\

& CGKF
& \PS{csvs/metrics_LC-Gauss.csv}{CGKF-sw}{\SelectedSNR}
& \RS{csvs/metrics_LC-Gauss.csv}{CGKF-sw}{\SelectedSNR}
& \PT{csvs/metrics_LC-Gauss.csv}{CGKF-TjW}{\SelectedSNR}
& \RSB{csvs/metrics_LC-Gauss.csv}{CGKF-TjW}{\SelectedSNR}
\\

& CQR
& \PS{csvs/metrics_LC-Gauss.csv}{CQR-sw}{\SelectedSNR}
& \RS{csvs/metrics_LC-Gauss.csv}{CQR-sw}{\SelectedSNR}
& \PT{csvs/metrics_LC-Gauss.csv}{CQR-Tjw}{\SelectedSNR}
& \RS{csvs/metrics_LC-Gauss.csv}{CQR-Tjw}{\SelectedSNR}
\\

& CQKF
& \PS{csvs/metrics_LC-Gauss.csv}{CQKF-sw}{\SelectedSNR}
& \RS{csvs/metrics_LC-Gauss.csv}{CQKF-sw}{\SelectedSNR}
& \PT{csvs/metrics_LC-Gauss.csv}{CQKF-Tjw}{\SelectedSNR}
& \RS{csvs/metrics_LC-Gauss.csv}{CQKF-Tjw}{\SelectedSNR}
\\

& Residuals-LCP
& --
& --
& \PT{csvs/metrics_LC-Gauss.csv}{LCP-TS}{\SelectedSNR}
& \RS{csvs/metrics_LC-Gauss.csv}{LCP-TS}{\SelectedSNR}
\\

& Residuals
& --
& --
& \PT{csvs/metrics_LC-Gauss.csv}{Union-TS}{\SelectedSNR}
& \RS{csvs/metrics_LC-Gauss.csv}{Union-TS}{\SelectedSNR}
\\

& Gauss-Bonf*
& --
& --
& \PT{csvs/metrics_LC-Gauss.csv}{KF-Gauss-Bonf}{\SelectedSNR}
& \RS{csvs/metrics_LC-Gauss.csv}{KF-Gauss-Bonf}{\SelectedSNR}
\\

& CGKF-Bonf*
& --
& --
& \PT{csvs/metrics_LC-Gauss.csv}{CGKF-Bonf-sw}{\SelectedSNR}
& \RS{csvs/metrics_LC-Gauss.csv}{CGKF-Bonf-sw}{\SelectedSNR}
\\

\midrule


\multirow{8}{*}{\textbf{Lin. Non-Gauss.}}

& Gauss
& \PS{csvs/metrics_LC-non-Gauss.csv}{KF-Gauss}{\SelectedSNR}
& \RS{csvs/metrics_LC-non-Gauss.csv}{KF-Gauss}{\SelectedSNR}
& --
& --
\\

& CGKF
& \PS{csvs/metrics_LC-non-Gauss.csv}{CGKF-sw}{\SelectedSNR}
& \RSB{csvs/metrics_LC-non-Gauss.csv}{CGKF-sw}{\SelectedSNR}
& \PT{csvs/metrics_LC-non-Gauss.csv}{CGKF-TjW}{\SelectedSNR}
& \RSB{csvs/metrics_LC-non-Gauss.csv}{CGKF-TjW}{\SelectedSNR}
\\

& CQR
& \PS{csvs/metrics_LC-non-Gauss.csv}{CQR-sw}{\SelectedSNR}
& \RS{csvs/metrics_LC-non-Gauss.csv}{CQR-sw}{\SelectedSNR}
& \PT{csvs/metrics_LC-non-Gauss.csv}{CQR-Tjw}{\SelectedSNR}
& \RS{csvs/metrics_LC-non-Gauss.csv}{CQR-Tjw}{\SelectedSNR}
\\

& CQKF
& \PS{csvs/metrics_LC-non-Gauss.csv}{CQKF-sw}{\SelectedSNR}
& \RS{csvs/metrics_LC-non-Gauss.csv}{CQKF-sw}{\SelectedSNR}
& \PT{csvs/metrics_LC-non-Gauss.csv}{CQKF-Tjw}{\SelectedSNR}
& \RS{csvs/metrics_LC-non-Gauss.csv}{CQKF-Tjw}{\SelectedSNR}
\\

& Residuals-LCP
& --
& --
& \PT{csvs/metrics_LC-non-Gauss.csv}{LCP-TS}{\SelectedSNR}
& \RS{csvs/metrics_LC-non-Gauss.csv}{LCP-TS}{\SelectedSNR}
\\

& Residuals
& --
& --
& \PT{csvs/metrics_LC-non-Gauss.csv}{Union-TS}{\SelectedSNR}
& \RS{csvs/metrics_LC-non-Gauss.csv}{Union-TS}{\SelectedSNR}
\\

& Gauss-Bonf*
& --
& --
& \textcolor{red}{%
    \PT{csvs/metrics_LC-non-Gauss.csv}{KF-Gauss-Bonf}{\SelectedSNR}}
& \RS{csvs/metrics_LC-non-Gauss.csv}{KF-Gauss-Bonf}{\SelectedSNR}
\\

& CGKF-Bonf*
& --
& --
& \PT{csvs/metrics_LC-non-Gauss.csv}{CGKF-Bonf-sw}{\SelectedSNR}
& \RS{csvs/metrics_LC-non-Gauss.csv}{CGKF-Bonf-sw}{\SelectedSNR}
\\

\midrule


\multirow{8}{*}{\textbf{Non-linear (Toy)}}

& Gauss
& \textcolor{red}{%
    \PS{csvs/metrics_non-Linear.csv}{KF-Gauss}{\SelectedSNR}}
& \RS{csvs/metrics_non-Linear.csv}{KF-Gauss}{\SelectedSNR}
& --
& --
\\

& CGKF
& \PS{csvs/metrics_non-Linear.csv}{CGKF-sw}{\SelectedSNR}
& \RS{csvs/metrics_non-Linear.csv}{CGKF-sw}{\SelectedSNR}
& \PT{csvs/metrics_non-Linear.csv}{CGKF-TjW}{\SelectedSNR}
& \RS{csvs/metrics_non-Linear.csv}{CGKF-TjW}{\SelectedSNR}
\\

& CQR
& \PS{csvs/metrics_non-Linear.csv}{CQR-sw}{\SelectedSNR}
& \RS{csvs/metrics_non-Linear.csv}{CQR-sw}{\SelectedSNR}
& \PT{csvs/metrics_non-Linear.csv}{CQR-Tjw}{\SelectedSNR}
& \RS{csvs/metrics_non-Linear.csv}{CQR-Tjw}{\SelectedSNR}
\\

& CQKF
& \PS{csvs/metrics_non-Linear.csv}{CQKF-sw}{\SelectedSNR}
& \RSB{csvs/metrics_non-Linear.csv}{CQKF-sw}{\SelectedSNR}
& \PT{csvs/metrics_non-Linear.csv}{CQKF-Tjw}{\SelectedSNR}
& \RSB{csvs/metrics_non-Linear.csv}{CQKF-Tjw}{\SelectedSNR}
\\

& Residuals-LCP
& --
& --
& \PT{csvs/metrics_non-Linear.csv}{LCP-TS}{\SelectedSNR}
& \RS{csvs/metrics_non-Linear.csv}{LCP-TS}{\SelectedSNR}
\\

& Residuals
& --
& --
& \PT{csvs/metrics_non-Linear.csv}{Union-TS}{\SelectedSNR}
& \RS{csvs/metrics_non-Linear.csv}{Union-TS}{\SelectedSNR}
\\

& Gauss-Bonf*
& --
& --
& \textcolor{red}{%
    \PT{csvs/metrics_non-Linear.csv}{KF-Gauss-Bonf}{\SelectedSNR}}
& \RS{csvs/metrics_non-Linear.csv}{KF-Gauss-Bonf}{\SelectedSNR}
\\

& CGKF-Bonf*
& --
& --
& \PT{csvs/metrics_non-Linear.csv}{CGKF-Bonf-sw}{\SelectedSNR}
& \RS{csvs/metrics_non-Linear.csv}{CGKF-Bonf-sw}{\SelectedSNR}
\\

\midrule

\multirow{8}{*}{\textbf{Non-linear (Toy) $\&$ Mismatch}}

& Gauss
& \textcolor{red}{%
    \PS{csvs/metrics_non-Linear-partial.csv}{KF-Gauss}{\SelectedSNR}}
& \RS{csvs/metrics_non-Linear-partial.csv}{KF-Gauss}{\SelectedSNR}
& --
& --
\\

& CGKF
& \PS{csvs/metrics_non-Linear-partial.csv}{CGKF-sw}{\SelectedSNR}
& \RS{csvs/metrics_non-Linear-partial.csv}{CGKF-sw}{\SelectedSNR}
& \PT{csvs/metrics_non-Linear-partial.csv}{CGKF-TjW}{\SelectedSNR}
& \RS{csvs/metrics_non-Linear-partial.csv}{CGKF-TjW}{\SelectedSNR}
\\

& CQR
& \PS{csvs/metrics_non-Linear-partial.csv}{CQR-sw}{\SelectedSNR}
& \RS{csvs/metrics_non-Linear-partial.csv}{CQR-sw}{\SelectedSNR}
& \PT{csvs/metrics_non-Linear-partial.csv}{CQR-Tjw}{\SelectedSNR}
& \RS{csvs/metrics_non-Linear-partial.csv}{CQR-Tjw}{\SelectedSNR}
\\

& CQKF
& \PS{csvs/metrics_non-Linear-partial.csv}{CQKF-sw}{\SelectedSNR}
& \RSB{csvs/metrics_non-Linear-partial.csv}{CQKF-sw}{\SelectedSNR}
& \PT{csvs/metrics_non-Linear-partial.csv}{CQKF-Tjw}{\SelectedSNR}
& \RSB{csvs/metrics_non-Linear-partial.csv}{CQKF-Tjw}{\SelectedSNR}
\\

& Residuals-LCP
& --
& --
& \PT{csvs/metrics_non-Linear-partial.csv}{LCP-TS}{\SelectedSNR}
& \RS{csvs/metrics_non-Linear-partial.csv}{LCP-TS}{\SelectedSNR}
\\

& Residuals
& --
& --
& \PT{csvs/metrics_non-Linear-partial.csv}{Union-TS}{\SelectedSNR}
& \RS{csvs/metrics_non-Linear-partial.csv}{Union-TS}{\SelectedSNR}
\\

& Gauss-Bonf*
& --
& --
& \textcolor{red}{%
    \PT{csvs/metrics_non-Linear-partial.csv}{KF-Gauss-Bonf}{\SelectedSNR}}
& \RS{csvs/metrics_non-Linear-partial.csv}{KF-Gauss-Bonf}{\SelectedSNR}
\\

& CGKF-Bonf*
& --
& --
& \PT{csvs/metrics_non-Linear-partial.csv}{CGKF-Bonf-sw}{\SelectedSNR}
& \RS{csvs/metrics_non-Linear-partial.csv}{CGKF-Bonf-sw}{\SelectedSNR}
\\

\bottomrule
\end{tabular}

\end{adjustbox}

\caption{
Performance comparison for scalar scenarios at $\mathrm{SNR}=0$~dB.
Each normalized interval-width entry reports the average value/STD.
For sample-wise evaluation, both quantities are normalized by the mean width
of the corresponding Gauss baseline. For trajectory-wise evaluation, they are
normalized by the mean width of the corresponding Gauss-Bonf baseline.
A value greater than one indicates an interval wider than the corresponding
Gaussian-posterior baseline.
Methods marked with * use an extended calibration set.
\textcolor{red}{Red} entries denote failure to meet the required
$\leq 5\%$ miscoverage.
A dash indicates that the corresponding guarantee is not provided
by the method.
}
\label{tab:results_1D}

\end{table*}

\smallskip
\noindent
{\bf Compared Methods:}
We evaluate the proposed \ac{cgkf}, \ac{cqkf}, and \ac{cdkf}. As an uncalibrated model-based reference, we use the Gaussian posterior reported by the Kalman-type filter, denoted \textbf{Gauss} for both scalar and multivariate states.
For trajectory-wise scalar evaluation, we also consider the Bonferroni variants \textbf{Gauss-Bonf} and \textbf{\ac{cgkf}-Bonf}, which use the sample-wise level $\tilde{\alpha}=\alpha/T$.

To compare with \ac{cp} methods that do not exploit the Kalman posterior, we use standard \ac{cqr} for scalar states, and \ac{dqr} and \ac{dcp} for multivariate states. These employ the observations as their predictive features, whereas \ac{cqkf} and \ac{cdkf} use the posterior moments $(\hat{\myVec{s}}_t,\myMat{\Sigma}_t)$. Thus, these comparisons isolate the benefit of incorporating the statistical information produced by the Kalman-type filter. For scalar trajectory-wise confidence, we additionally compare with the residual-based method of~\cite{cleaveland2024conformal}, using both its learned time-dependent weighting, denoted \textbf{Residuals-LCP}, and its unweighted counterpart, denoted \textbf{Residuals}. For multivariate states we further include \textbf{Rec}, which applies scalar \ac{cgkf} independently to each state coordinate and forms an axis-aligned rectangular region; this serves as a simple reference that does not exploit cross-coordinate dependencies.


\pgfplotstableread[col sep=comma]{csvs/metrics_linear.csv}\LinearData
\pgfplotstableread[col sep=comma]{csvs/metrics_nonlinear.csv}\NonlinearData
\pgfplotstableread[col sep=comma]{csvs/metrics_lorenz.csv}\LorenzData
\pgfplotstableread[col sep=comma]{csvs/metrics_nonlinear-UKF.csv}\NonlinearUKFData


\pgfplotstableread[col sep=comma]{csvs/metrics_linear_traj.csv}\LinearTrajData
\pgfplotstableread[col sep=comma]{csvs/metrics_nonlinear_traj.csv}\NonlinearTrajData
\pgfplotstableread[col sep=comma]{csvs/metrics_nonlinear_UKF_traj.csv}\NonlinearUKFTrajData
\pgfplotstableread[col sep=comma]{csvs/metrics_lorenz_traj.csv}\LorenzTrajData


\newcommand{\CSVnum}[3]{%
    \pgfplotstablegetelem{#2}{#3}\of#1%
    \pgfmathprintnumber[
        fixed,
        fixed zerofill,
        precision=2
    ]{\pgfplotsretval}%
}

%

\newcommand{\CSVregion}[2]{%
    \pgfplotstablegetelem{#2}{WI_mean}\of#1%
    \edef\RegionMean{\pgfplotsretval}%

    \pgfplotstablegetelem{#2}{WI_svd}\of#1%
    \edef\RegionStd{\pgfplotsretval}%

    \pgfplotstablegetelem{3}{WI_mean}\of#1%
    \edef\GaussianMean{\pgfplotsretval}%

    \pgfmathparse{\RegionMean/\GaussianMean}%
    \edef\NormalizedRegionMean{\pgfmathresult}%

    \pgfmathparse{\RegionStd/\GaussianMean}%
    \edef\NormalizedRegionStd{\pgfmathresult}%

    \pgfmathprintnumber[
        fixed,
        fixed zerofill,
        precision=2
    ]{\NormalizedRegionMean}%
    /%
    \pgfmathprintnumber[
        fixed,
        fixed zerofill,
        precision=2
    ]{\NormalizedRegionStd}%
}


\newcommand{\CSVRegionBold}[2]{%
    \pgfplotstablegetelem{#2}{WI_mean}\of#1%
    \edef\RegionMean{\pgfplotsretval}%

    \pgfplotstablegetelem{#2}{WI_svd}\of#1%
    \edef\RegionStd{\pgfplotsretval}%

    \pgfplotstablegetelem{3}{WI_mean}\of#1%
    \edef\GaussianMean{\pgfplotsretval}%

    \pgfmathparse{\RegionMean/\GaussianMean}%
    \edef\NormalizedRegionMean{\pgfmathresult}%

    \pgfmathparse{\RegionStd/\GaussianMean}%
    \edef\NormalizedRegionStd{\pgfmathresult}%

    {\bfseries\boldmath
    \pgfmathprintnumber[
        fixed,
        fixed zerofill,
        precision=2
    ]{\NormalizedRegionMean}%
    /%
    \pgfmathprintnumber[
        fixed,
        fixed zerofill,
        precision=2
    ]{\NormalizedRegionStd}%
    }%
}


\begin{table*}[!t]
\centering

\setlength{\tabcolsep}{5pt}
\renewcommand{\arraystretch}{1.05}
\footnotesize

\begin{adjustbox}{max width=\textwidth}

\begin{tabular}{llcc|cc}
\toprule

\multirow{2}{*}{\textbf{Scenario}} &
\multirow{2}{*}{\textbf{Method}} &
\multicolumn{2}{c|}{\textbf{Sample-wise}} &
\multicolumn{2}{c}{\textbf{Trajectory-wise}}
\\

\cmidrule(lr){3-4}
\cmidrule(lr){5-6}

&
&
\textbf{Per-sample Miscoverage [\%]} &
\textbf{Normalized Region Size} &
\textbf{Per-trajectory Miscoverage [\%]} &
\textbf{Normalized Region Size}
\\

\midrule


\multirow{7}{*}{\textbf{Linear Gaussian}}

& Gauss
& \CSVnum{\LinearData}{3}{C_mean}/\CSVnum{\LinearData}{3}{C_svd}
& \CSVregion{\LinearData}{3}
& --
& --
\\

& CGKF
& \CSVnum{\LinearData}{2}{C_mean}/\CSVnum{\LinearData}{2}{C_svd}
& \CSVRegionBold{\LinearData}{2}
& \CSVnum{\LinearTrajData}{2}{S}
& \CSVRegionBold{\LinearTrajData}{2}
\\

& Rec
& \CSVnum{\LinearData}{4}{C_mean}/\CSVnum{\LinearData}{4}{C_svd}
& \CSVregion{\LinearData}{4}
& \CSVnum{\LinearTrajData}{4}{S}
& \CSVregion{\LinearTrajData}{4}
\\

& DQR
& \CSVnum{\LinearData}{1}{C_mean}/\CSVnum{\LinearData}{1}{C_svd}
& \CSVregion{\LinearData}{1}
& \CSVnum{\LinearTrajData}{1}{S}
& \CSVregion{\LinearTrajData}{1}
\\

& CQKF
& \CSVnum{\LinearData}{0}{C_mean}/\CSVnum{\LinearData}{0}{C_svd}
& \CSVregion{\LinearData}{0}
& \CSVnum{\LinearTrajData}{0}{S}
& \CSVregion{\LinearTrajData}{0}
\\

& DCP
& \CSVnum{\LinearData}{6}{C_mean}/\CSVnum{\LinearData}{6}{C_svd}
& \CSVregion{\LinearData}{6}
& \CSVnum{\LinearTrajData}{6}{S}
& \CSVregion{\LinearTrajData}{6}
\\

& CDKF
& \CSVnum{\LinearData}{5}{C_mean}/\CSVnum{\LinearData}{5}{C_svd}
& \CSVregion{\LinearData}{5}
& \CSVnum{\LinearTrajData}{5}{S}
& \CSVregion{\LinearTrajData}{5}
\\

\midrule


\multirow{7}{*}{\textbf{Pendulum}}

& Gauss
& \textcolor{red}{%
    \CSVnum{\NonlinearData}{3}{C_mean}/%
    \CSVnum{\NonlinearData}{3}{C_svd}}
& \CSVregion{\NonlinearData}{3}
& --
& --
\\

& CGKF
& \CSVnum{\NonlinearData}{2}{C_mean}/\CSVnum{\NonlinearData}{2}{C_svd}
& \CSVregion{\NonlinearData}{2}
& \CSVnum{\NonlinearTrajData}{2}{S}
& \CSVregion{\NonlinearTrajData}{2}
\\

& Rec
& \CSVnum{\NonlinearData}{4}{C_mean}/\CSVnum{\NonlinearData}{4}{C_svd}
& \CSVregion{\NonlinearData}{4}
& \CSVnum{\NonlinearTrajData}{4}{S}
& \CSVregion{\NonlinearTrajData}{4}
\\

& DQR
& \CSVnum{\NonlinearData}{1}{C_mean}/\CSVnum{\NonlinearData}{1}{C_svd}
& \CSVregion{\NonlinearData}{1}
& \CSVnum{\NonlinearTrajData}{1}{S}
& \CSVregion{\NonlinearTrajData}{1}
\\

& CQKF
& \CSVnum{\NonlinearData}{0}{C_mean}/\CSVnum{\NonlinearData}{0}{C_svd}
& \CSVregion{\NonlinearData}{0}
& \CSVnum{\NonlinearTrajData}{0}{S}
& \CSVregion{\NonlinearTrajData}{0}
\\

& DCP
& \CSVnum{\NonlinearData}{6}{C_mean}/\CSVnum{\NonlinearData}{6}{C_svd}
& \CSVregion{\NonlinearData}{6}
& \CSVnum{\NonlinearTrajData}{6}{S}
& \CSVregion{\NonlinearTrajData}{6}
\\

& CDKF
& \CSVnum{\NonlinearData}{5}{C_mean}/\CSVnum{\NonlinearData}{5}{C_svd}
& \CSVRegionBold{\NonlinearData}{5}
& \CSVnum{\NonlinearTrajData}{5}{S}
& \CSVRegionBold{\NonlinearTrajData}{5}
\\

\midrule


\multirow{7}{*}{\textbf{Pendulum (UKF)}}

& Gauss
& \textcolor{red}{%
    \CSVnum{\NonlinearUKFData}{3}{C_mean}/%
    \CSVnum{\NonlinearUKFData}{3}{C_svd}}
& \CSVregion{\NonlinearUKFData}{3}
& --
& --
\\

& CGKF
& \CSVnum{\NonlinearUKFData}{2}{C_mean}/\CSVnum{\NonlinearUKFData}{2}{C_svd}
& \CSVregion{\NonlinearUKFData}{2}
& \CSVnum{\NonlinearUKFTrajData}{2}{S}
& \CSVregion{\NonlinearUKFTrajData}{2}
\\

& Rec
& \CSVnum{\NonlinearUKFData}{4}{C_mean}/\CSVnum{\NonlinearUKFData}{4}{C_svd}
& \CSVregion{\NonlinearUKFData}{4}
& \CSVnum{\NonlinearUKFTrajData}{4}{S}
& \CSVregion{\NonlinearUKFTrajData}{4}
\\

& DQR
& \CSVnum{\NonlinearUKFData}{1}{C_mean}/\CSVnum{\NonlinearUKFData}{1}{C_svd}
& \CSVregion{\NonlinearUKFData}{1}
& \CSVnum{\NonlinearUKFTrajData}{1}{S}
& \CSVregion{\NonlinearUKFTrajData}{1}
\\

& CQKF
& \CSVnum{\NonlinearUKFData}{0}{C_mean}/\CSVnum{\NonlinearUKFData}{0}{C_svd}
& \CSVregion{\NonlinearUKFData}{0}
& \CSVnum{\NonlinearUKFTrajData}{0}{S}
& \CSVregion{\NonlinearUKFTrajData}{0}
\\

& DCP
& \CSVnum{\NonlinearUKFData}{6}{C_mean}/\CSVnum{\NonlinearUKFData}{6}{C_svd}
& \CSVregion{\NonlinearUKFData}{6}
& \CSVnum{\NonlinearUKFTrajData}{6}{S}
& \CSVregion{\NonlinearUKFTrajData}{6}
\\

& CDKF
& \CSVnum{\NonlinearUKFData}{5}{C_mean}/\CSVnum{\NonlinearUKFData}{5}{C_svd}
& \CSVRegionBold{\NonlinearUKFData}{5}
& \CSVnum{\NonlinearUKFTrajData}{5}{S}
& \CSVRegionBold{\NonlinearUKFTrajData}{5}
\\

\midrule


\multirow{7}{*}{\textbf{Lorenz}}

& Gauss
& \textcolor{red}{%
    \CSVnum{\LorenzData}{3}{C_mean}/%
    \CSVnum{\LorenzData}{3}{C_svd}}
& \CSVregion{\LorenzData}{3}
& --
& --
\\

& CGKF
& \CSVnum{\LorenzData}{2}{C_mean}/\CSVnum{\LorenzData}{2}{C_svd}
& \CSVregion{\LorenzData}{2}
& \CSVnum{\LorenzTrajData}{2}{S}
& \CSVregion{\LorenzTrajData}{2}
\\

& Rec
& \CSVnum{\LorenzData}{4}{C_mean}/\CSVnum{\LorenzData}{4}{C_svd}
& \CSVregion{\LorenzData}{4}
& \CSVnum{\LorenzTrajData}{4}{S}
& \CSVregion{\LorenzTrajData}{4}
\\

& DQR
& \CSVnum{\LorenzData}{1}{C_mean}/\CSVnum{\LorenzData}{1}{C_svd}
& \CSVregion{\LorenzData}{1}
& \CSVnum{\LorenzTrajData}{1}{S}
& \CSVregion{\LorenzTrajData}{1}
\\

& CQKF
& \CSVnum{\LorenzData}{0}{C_mean}/\CSVnum{\LorenzData}{0}{C_svd}
& \CSVregion{\LorenzData}{0}
& \CSVnum{\LorenzTrajData}{0}{S}
& \CSVRegionBold{\LorenzTrajData}{0}
\\

& DCP
& \CSVnum{\LorenzData}{6}{C_mean}/\CSVnum{\LorenzData}{6}{C_svd}
& \CSVregion{\LorenzData}{6}
& \CSVnum{\LorenzTrajData}{6}{S}
& \CSVregion{\LorenzTrajData}{6}
\\

& CDKF
& \CSVnum{\LorenzData}{5}{C_mean}/\CSVnum{\LorenzData}{5}{C_svd}
& \CSVRegionBold{\LorenzData}{5}
& \CSVnum{\LorenzTrajData}{5}{S}
& \CSVregion{\LorenzTrajData}{5}
\\

\bottomrule
\end{tabular}

\end{adjustbox}

\caption{
Performance comparison for multivariate scenarios at
$\mathrm{SNR}=-15$~dB.
Each miscoverage entry reports the average value/STD, while each
normalized region-size entry reports the normalized average/STD.
For both sample-wise and trajectory-wise evaluation, the region-size
mean and STD are normalized by the mean region size of the corresponding
Gaussian-posterior baseline.
Thus, a normalized value greater than one indicates a region larger than
the corresponding Gaussian-posterior region.
\textcolor{red}{Red} entries denote failure to meet the required
$\leq 5\%$ miscoverage rate.
}

\label{tab:multidim_results}

\end{table*}

\smallskip
\noindent
{\bf Learning and Evaluation:}
For \ac{dqr}-based methods we use an \ac{mlp} with hidden-layer widths $[128,128]$ and $|\mySet{U}|=128$ unit-norm directions, sampled by drawing $\tilde{\myVec{u}}_j\sim\mathcal{N}(\mathbf{0},\myMat{I})$ and setting $\myVec{u}_j=\tilde{\myVec{u}}_j/\|\tilde{\myVec{u}}_j\|_2$. Following~\cite{feldman2023calibrated}, the directional miscoverage levels used during training are set to $\alpha/7$ and $\alpha/20$ in two and three dimensions, respectively; conformal calibration itself always uses the nominal level $\alpha$. For \ac{dcp}-based methods, we use an \ac{mlp} with hidden-layer widths $[100,100]$ to parameterize a mixture of $M=10$ Gaussian components. The \ac{dqr} networks are trained with the pinball loss~\eqref{eqn:Pinball}, whereas the \ac{dcp} networks are trained by negative log-likelihood.

We set the miscoverage level to $\alpha=5\%$. The scalar experiments use trajectories of length $T=100$ and are averaged over 200 trials, each defined by a random $80\%$–$20\%$ calibration–test split, at $\mathrm{SNR}=0$~dB.The Bonferroni variants use a dataset of size $|\mathcal{D}|=10{,}000$, whereas the remaining methods use $|\mathcal{D}|=1{,}000$. For \textbf{Residuals-LCP}, \ac{cqr}, and \ac{cqkf}, $10\%$ of these $1{,}000$ samples are used for learning and the remaining $90\%$ for conformal calibration. The \ac{cqr} and \ac{cqkf} networks are trained for 500 epochs.

The multivariate experiments use trajectories of length $T=50$. Learning-based methods use 1000 trajectories for training, 800 for calibration, and 200 independent trajectories for testing, and are trained for 1000 epochs. The results are averaged over 20 random splits of the available trajectories into calibration, and test sets. Methods that require no training use the combined training and calibration data, i.e., 1800 trajectories, for  calibration so as to keep the total amount of available data fixed.

We report both \emph{per-sample miscoverage} and \emph{per-trajectory miscoverage}. The former is the fraction of individual states lying outside their corresponding regions, while the latter is the fraction of trajectories containing at least one uncovered state. Region efficiency is measured by interval width in the scalar case and by region volume for multivariate states. The latter is estimated using a uniform grid whose limits are determined by the empirical $0.01$ and $0.99$ quantiles of the training targets with a margin of $2$. If $N_{\rm in}$ of $N_{\rm grid}$ grid points belong to the confidence region, we approximate its volume as the $\frac{N_{\rm in}}{N_{\rm grid}}$ portion of the grid volume.  We report the  values for the interval width/region size normalized by the corresponding values obtained with standard Gaussian modeling.

\subsection{Scalar-State Results}
\label{subsec:scalar_results}
The scalar-state results are summarized in Table~\ref{tab:results_1D}.
The uncalibrated \textbf{Gauss} intervals attain the desired coverage on average in the matched linear Gaussian and linear non-Gaussian settings but exhibit high variability, even in the linear Gaussian setting, and become unreliable when the model assumptions are violated.
The conformally calibrated methods attain the prescribed sample-wise or trajectory-wise reliability across the considered scenarios. As expected from the analysis in Subsection~\ref{subsec: Guarantees},
the trajectory-wise constructions achieve the corresponding joint coverage.

Among the proposed methods, \ac{cgkf} yields particularly tight intervals in the matched linear Gaussian setting, where the Gaussian geometry used in its initialization accurately represents the posterior. In the nonlinear setting, \ac{cqkf} benefits from its ability to construct asymmetric intervals and yields tighter regions, as illustrated in Fig.~\ref{fig:CGKF_CQKF}. Moreover, even under model mismatch, \ac{cqkf} produces narrower intervals than observation-based \ac{cqr}, demonstrating the benefit of using the time-varying posterior moments as conformal features. The direct trajectory-wise constructions also avoid the substantially larger calibration sets required by the Bonferroni variants. The proposed methods also outperform the residual-based baselines, \textbf{Residuals} is disadvantaged by a non-adaptive, constant-width band over time \(t\), while \textbf{Residuals-LCP} further loses efficiency by reserving \(10\%\) of the data to train its coefficients.

\begin{figure*}[t!]
    \centering
    \includegraphics[
        width=\textwidth,
        trim=50 30 50 30,
        clip
    ]{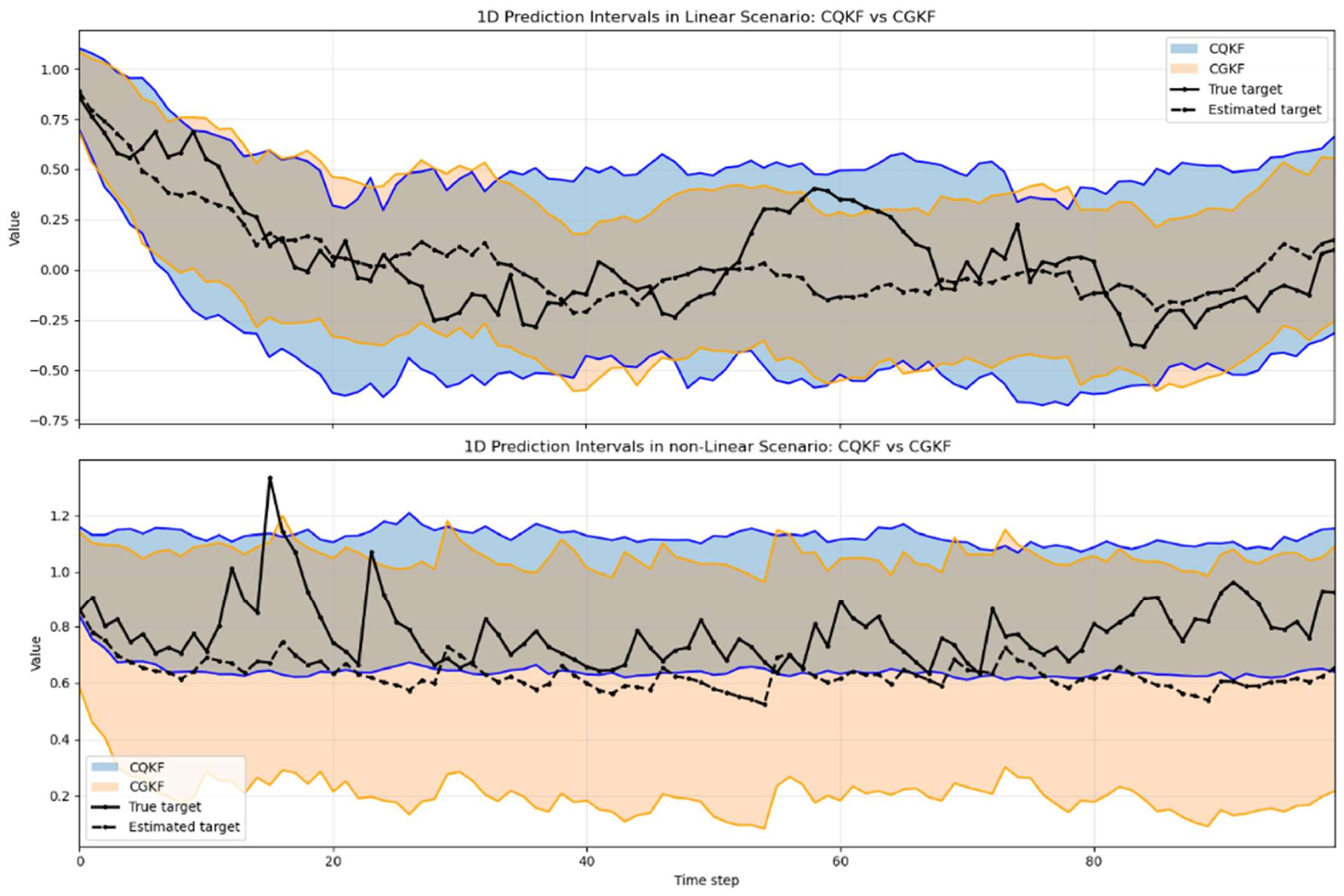}
    \caption{One-dimensional prediction intervals obtained by
    \ac{cqkf} and \ac{cgkf} for the linear and nonlinear
    \ac{ss} models.}
    \label{fig:CGKF_CQKF}
\end{figure*}

\subsection{Multivariate-State Results}
\label{subsec:multidim_results}
The multivariate results are summarized in Table~\ref{tab:multidim_results}. The results demonstrate the importance of conformal calibration: while the nominal Gaussian posterior can provide reasonable sample-wise regions in the matched linear Gaussian setting, its coverage deteriorates markedly for nonlinear and chaotic dynamics. All three proposed constructions largely restore the prescribed coverage, for both sample-wise and trajectory-wise calibration, despite their markedly different region geometries.

The results further demonstrate the benefit of exploiting the posterior information generated by the Kalman-type filter. The observation-based \ac{dqr} and \ac{dcp} baselines generally require larger regions than their posterior-aware counterparts \ac{cqkf} and \ac{cdkf}. In the matched linear Gaussian setting, \ac{cgkf} provides the most compact regions, as expected from the agreement between its ellipsoidal construction and the true posterior geometry. As the posterior departs from this setting, the learned constructions become increasingly advantageous. In particular, \ac{cdkf} generally produces the smallest regions among the learned approaches, while \ac{cqkf} also provides a substantial reduction relative to conventional \ac{dqr}. The improvement is especially pronounced for the Lorenz attractor, where the greater flexibility of the learned posterior-aware constructions yields considerably more compact regions than either Gaussian or observation-only conformal alternatives.

\subsection{Complexity and Sensitivity}
\label{subsec:complexity_results}

We conclude by examining the computational cost associated with the additional flexibility of our constructions. Since model training and conformal calibration are carried out offline, inference latency is particularly relevant for online tracking. The resulting latencies for the Lorenz experiment, all measured on the same AMD Ryzen AI 9 HX 370 CPU platform, are reported in Table~\ref{tab:complexity}.
As expected, \ac{cgkf} incurs essentially no learned-model overhead beyond the Kalman-type filtering operation and the evaluation of the Mahalanobis score. \ac{cqkf} introduces the evaluation of a directional quantile network but remains substantially less computationally demanding than \ac{cdkf}. The latter provides the most flexible confidence regions, but its mixture-density representation leads to considerably higher inference times. These results quantify the geometry-complexity tradeoff discussed in Subsection~\ref{subsec: Discussion}.




\begin{table}[t]
\centering
\caption{Average inference time per sample for the considered methods in the Lorenz attractor simulation.}
\label{tab:complexity}
\setlength{\tabcolsep}{3.5pt}
\resizebox{\linewidth}{!}{%
\begin{tabular}{lccccccc}
\toprule
Method
& Gauss
& CGKF
& Rec
& DQR
& CQKF
& DCP
& CDKF
\\
\midrule
Evaluation [ms]
& 0.149
& 0.155
& 1.075
& 0.870
& 0.878
& 1.830
& 1.862
\\
\\
\bottomrule
\end{tabular}
}
\end{table}
Throughout the multivariate experiments reported above, we use $|\mySet{U}|=128$.
The computational cost of \ac{cqkf} can further be adjusted through the number of directional constraints $|\mySet{U}|$. Fig.~\ref{fig:effect_number_of_directions} evaluates this tradeoff for the linear Gaussian model at $\mathrm{SNR}=-15$~dB. Increasing $|\mySet{U}|$ yields a finer approximation of the region boundary and reduces the resulting region size, bringing it closer to the ellipsoidal \ac{cgkf} construction in this setting, at the expense of evaluating more directional quantiles. 

\begin{figure}[t!]
\centering
\includegraphics[width=\columnwidth]{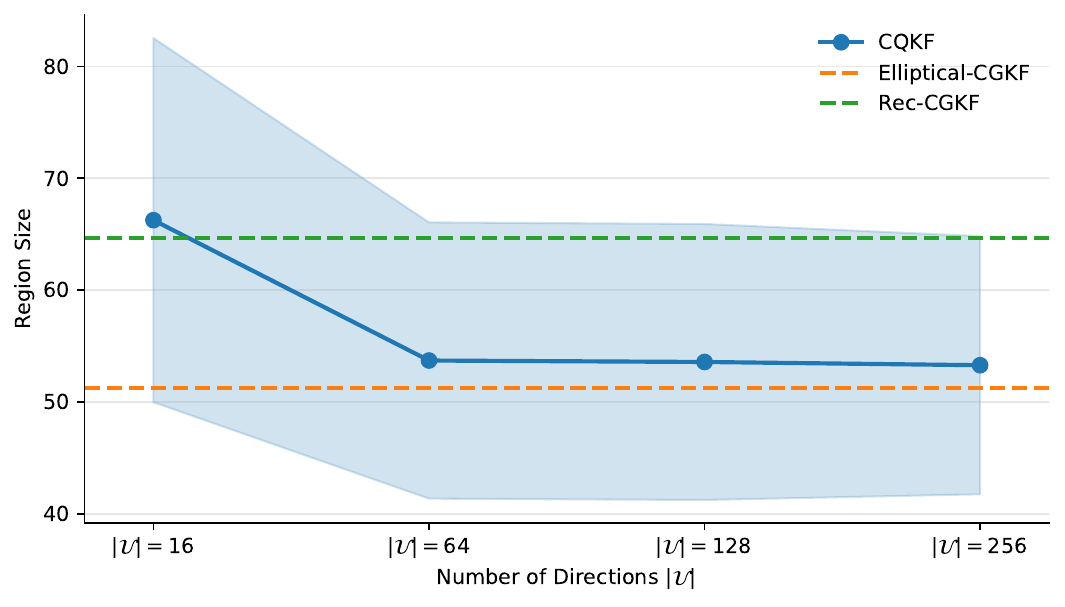}
\caption{Effect of the number of directions $|\mySet{U}|$ on the
prediction-region size for the linear Gaussian scenario at
$\mathrm{SNR}=-15$~dB. 
}
\label{fig:effect_number_of_directions}
\end{figure}

\section{Conclusion}
\label{sec:conclusion}
We proposed a conformalized Kalman filtering framework that leverages the time-varying posterior moments produced by Kalman-type estimators to construct statistically reliable confidence regions under nonlinearities, non-Gaussian disturbances, and model mismatch. The proposed \ac{cgkf}, \ac{cqkf}, and \ac{cdkf} constructions provide complementary tradeoffs between computational simplicity and region flexibility, while admitting finite-sample sample-wise and trajectory-wise coverage guarantees. Numerical results across diverse dynamical systems demonstrated that the proposed methods achieve the prescribed coverage while producing informative confidence regions, highlighting the benefit of combining model-based filtering with conformal calibration.

\bibliographystyle{IEEEtran}
\bibliography{IEEEabrv,refs}
\ifproofs
\appendix
\numberwithin{remark}{subsection} 
\numberwithin{equation}{subsection}

\subsection{Proof of Theorem~\ref{thm:c-kf-samplewise}}
\label{app:proof1} 

Fix $t\in\{1,\ldots,T\}$, and condition on the proper training set used by
the considered conformalized Kalman filtering method. Let
$R^{(1)}(t),\ldots,R^{|\mathcal{D}_{\mathrm{cal}}|}(t)$ denote the
nonconformity scores at time $t$ associated with the calibration
trajectories, and let $R^{\mathrm{test}}(t)$ denote the corresponding score
of the test trajectory.

By the construction of the confidence region, we have~\eqref{eqn:SmpWise},
\[
    s_t\in
    \mathcal{C}_t^\alpha(\hat{s}_t,\Sigma_t)
    \quad \text{iff} \quad
    R^{\mathrm{test}}(t)
    \leq
    \widehat{Q}^{1-\alpha}(t),
\]
where $\widehat{Q}^{1-\alpha}(t)$ is the computed correction from the calibration scores at time $t$.
Since the calibration and test trajectories are exchangeable, the
corresponding scores   $\{R^1(t),\ldots,R^{|\mathcal{D}_{\mathrm{cal}}|}(t),
    R^{\mathrm{test}}(t)
\}$
are also exchangeable. Therefore, by the empirical-quantile lemma for
exchangeable random variables~\cite{romano2019conformalized}, it holds that
\[
\Pr\!\left(
    R^{\mathrm{test}}(t)
    \leq
    \widehat{Q}^{1-\alpha}(t)
\right)
\geq 1-\alpha,
\]
and thus,
\[
    \Pr\!\left(
        s_t\in\mathcal{C}_t^\alpha(\hat{s}_t,\Sigma_t)
    \right)
    \geq 1-\alpha.
\]

Moreover, if the calibration and test nonconformity scores at time $t$ are
almost surely distinct, the corresponding upper bound of the
empirical-quantile lemma in~\cite{romano2019conformalized} gives
\[
\Pr\!\left(
    R^{\mathrm{test}}(t)
    \leq
    \widehat{Q}^{1-\alpha}(t)
\right)
\leq
1-\alpha+
\frac{1}{|\mathcal{D}_{\mathrm{cal}}|+1},
\]
and hence
\[
    \Pr\!\left(
        s_t\in\mathcal{C}_t^\alpha(\hat{s}_t,\Sigma_t)
    \right)
    \leq
    1-\alpha+
    \frac{1}{|\mathcal{D}_{\mathrm{cal}}|+1}.
\]


\subsection{Proof of Corollary ~\ref{cor:c-kf-bonferroni}}
\label{app:proof2}

The corollary follows from the union bound, as
\begin{align*}
\Pr\!\Big(\bigcup_{1\le t\le T}\{s_t \notin \mySet{C}_t^{\tilde{\alpha}}(\hat{s}_t, \Sigma_t)\}\Big)
&\;\le\; \sum_{1\le t\le T} \Pr\!\big(s_t \notin \mySet{C}_t^{\tilde{\alpha}}(\hat{s}_t, \Sigma_t)\big) \\
&\;\le\; \sum_{1\le t\le T} \tilde{\alpha} 
= \alpha.
\end{align*}
Taking complements proves \eqref{eqn:c-kf-bonferroni-coverage}.
\subsection{Proof of Theorem ~\ref{thm:c-kf-trajectorywise}}
\label{app:proof3} 

For any test trajectory, let $R_t$ denote the nonconformity score at time
$t$. By construction, the entire trajectory is covered if and only if all
sample-wise scores are below the calibrated threshold, or equivalently, if
their maximum is below this threshold. Thus,
\begin{align*}
&\Pr\!\big(s_t\in C_t^\alpha(\hat s_t,\Sigma_t)\ \forall t\in\{1,\ldots,T\}\big)
=\Pr\!\big(R_t\le \hat Q^{1-\alpha}\ \forall t\big) \\
&\qquad=\Pr\!\big(\bigcap_{t=1}^T \{R_t\le \hat Q^{1-\alpha}\}\big) \\
&\qquad=\Pr\!\big(\max_{1\le t\le T} R_t \le \hat Q^{1-\alpha}\big).
\end{align*}

Since the trajectories are exchangeable, their maximum nonconformity
scores are also exchangeable. The standard conformal empirical-quantile
argument therefore yields
\[
\Pr\!\big(\max_{1\le t\le T}R_t\le\hat Q^{1-\alpha}\big)
\geq 1-\alpha,
\]
which establishes the desired trajectory-wise coverage.


\fi
\end{document}